\documentclass{article}
\usepackage{amsmath,amsfonts,,amsthm, hyperref,algorithm,algorithmic,multirow,tabularx,paralist,nccmath, enumitem, fullpage}
\usepackage{natbib}
\usepackage[usenames,dvipsnames,svgnames,table]{xcolor}
\usepackage{tikz}

\hypersetup{
    bookmarks=true,         
    unicode=false,          
    pdftoolbar=true,        
    pdfmenubar=true,        
    pdffitwindow=false,     
    pdfstartview={FitH},    
    pdftitle={My title},    
    pdfauthor={Author},     
    pdfsubject={Subject},   
    pdfcreator={Creator},   
    pdfproducer={Producer}, 
    pdfkeywords={keyword1} {key2} {key3}, 
    pdfnewwindow=true,      
    colorlinks=true,       
    linkcolor=red,          
    citecolor=Blue,        
    filecolor=magenta,      
    urlcolor=cyan          
}

\newsavebox\MBox

\newtheorem{theorem}{Theorem}[section]
\newtheorem{lemma}[theorem]{Lemma}

\newtheorem{corollary}[theorem]{Corollary}

\newtheorem{definition}[theorem]{Definition}

\newtheorem{observation}[theorem]{Observation}

\newcommand{\secdiff}[2]{\frac{\partial^2 {#1}}{\partial {#2}^2}}

\newcommand{\E}[1]{{\bf{E}}\left[#1\right]}

\renewenvironment{proof}{\noindent{\bf Proof}:~}{$\hfill \Box$\\}

\def\b1{{\bf 1}}

\def\bx{{\bf x}}
\def\by{{\bf y}}
\def\bz{{\bf z}}
\def\bu{{\bf u}}

\def\RR{{\mathbb R}}

\def\cD{{\cal D}}

\def\cS{{\cal S}}

\def\eps {\epsilon}
\def\max{{\rm{max}}}

\title{Multiway Cut, Pairwise Realizable Distributions, \\ and Descending Thresholds}

\author{
Ankit Sharma\thanks{Carnegie Mellon University, Pittsburgh, PA; {\tt ankits@cs.cmu.edu}.
This work was done while the author was at IBM Almaden Research Center, San Jose, CA.}
\and Jan Vondr\'ak\thanks{IBM Almaden Research Center, San Jose, CA; {\tt jvondrak@us.ibm.com}.}
}
\date{\today}

\begin{document}
\maketitle
\begin{abstract}
We design new approximation algorithms for the Multiway Cut problem, improving the previously known factor of $1.32388$ \citep{BNS13}.

We proceed in three steps. First, we analyze the rounding scheme of \citet{BNS13} and design a modification that improves the approximation to $\frac{3+\sqrt{5}}{4} \approx 1.309017$. We also present a tight example showing that this is the best approximation one can achieve with the type of cuts considered by \citet{BNS13}: (1) partitioning by exponential clocks, and (2) single-coordinate cuts with equal thresholds.

Then, we prove that this factor can be improved by introducing a new rounding scheme: (3) single-coordinate cuts with descending thresholds. By combining these three schemes, we design an algorithm that achieves a factor of $\frac{10+4\sqrt{3}}{13} \approx 1.30217$. This is the best approximation factor that we are able to verify by hand.

Finally, we show that by combining these three rounding schemes with the scheme of independent thresholds from \cite{KKSTY04}, the approximation factor can be further improved to $1.2965$. This approximation factor has been verified only by computer.

\end{abstract}

\section{Introduction}
\label{sec:intro}

The Multiway Cut problem is one of the classical graph optimization problems: Given a graph $G = (V,E)$ with edge weights $w:E \rightarrow \RR_+$ and $k$ terminals $t_1,t_2,\ldots,t_k \in V$, we want to find a minimum-weight subset of edges $F \subseteq E$ such that no pair of terminals is connected in $(V, E \setminus F)$. Equivalently, we can search for a labeling of the vertices $\ell:V \rightarrow [k]$ so as to minimize the total weight of edges $(v,w)$ such that $\ell(v) \neq \ell(w)$. This is a natural generalization of the Minimum $s$-$t$-Cut problem, which is the $k=2$ case.

The study of Multiway Cut goes back to \citet{DJPSY94} who proved that the problem is MAX SNP-hard for every $k \geq 3$, and gave a simple combinatorial algorithm using repeated applications of Min $s$-$t$-Cut that achieves a $(2-2/k)$-approximation. A novel technique for Multiway Cut was introduced by \citet{CKR01} who proposed a geometric relaxation for this problem. In this ``CKR relaxation'', the graph is embedded into a simplex with each terminal at a distinct vertex, and finding the optimal embedding can be formulated as a linear program. A partitioning of the graph then corresponds to a partitioning of the simplex that separates all the vertices. Using this relaxation, \cite{CKR01} designed a $(1.5 - 1/k)$-approximation for Multiway Cut.

\citet{CT99}, and \citet{KKSTY04} independently provided a $12/11$-approximation for $k=3$ and showed that this is the best approximation achievable using the CKR relaxation for $k=3$, by presenting a matching integrality gap example. More generally, \citet{KKSTY04} provided improved approximation algorithms for all values of $k$, with approximation factors tending to $1.3438$ as $k \rightarrow \infty$. On the negative side, \citet{FK00} showed an integrality gap of $8/(7+1/(k-1))$ for each $k \geq 3$.

The importance of the CKR relaxation was bolstered further by the work of \citet{MNRS08} who proved that assuming the Unique Games Conjecture, it is NP-hard to achieve an approximation for Multiway Cut better than the integrality gap of the CKR relaxation (for any fixed $k$). This means that $12/11$ is indeed the best possible approximation for $k=3$, and for every $k$ the CKR relaxation provides the optimal approximation factor (assuming the UGC). 


Recently, \citet{BNS13} made a new improvement on the algorithmic side and designed a $1.32388$ approximation for Multiway Cut (for arbitrary $k$). They introduced an interesting new rounding scheme for the CKR relaxation that they called ``partitioning using exponential clocks''. 
(In fact, they also showed that a threshold-based rounding scheme from \citet{KT02} could be used equivalently in place of the exponential clocks.)  Combining this scheme with a modification of a thresholding scheme from \cite{CKR01}, they presented a very simple and elegant way to achieve a $4/3$-approximation. Then they modified the rounding scheme further, to improve the approximation factor to $1.32388$.

\medskip
\noindent{\bf Our contribution.}
We build upon previous work and provide further improvements on the approximation factor for Multiway Cut. First, we study the rounding scheme of \cite{BNS13} and identify a modification of their scheme that leads to a factor of $\frac{3+\sqrt{5}}{4} \approx 1.309017$. (See Section~\ref{sec:1.309}.) We note that this number is equal to $\frac{1+\varphi}{2}$ where $\varphi = \frac{1+\sqrt{5}}{2}$ is the {\em golden ratio}. We also present a tight example showing that this is the best approximation factor that can be achieved by any combination of the techniques considered by \cite{BNS13}: the exponential clocks scheme, the Kleinberg-Tardos scheme, and thresholding schemes with equal thresholds for all terminals. (We provide more details in Appendix~\ref{sec:tight-example}.)

Secondly, we improve this approximation factor by introducing a new rounding scheme that we call {\em descending thresholds}. This scheme can be combined with the rounding schemes above in a way that achieves an approximation factor of $\frac{10+4\sqrt{3}}{13} \approx 1.30217$. The analysis of this algorithm is still quite simple and can be verified easily by hand (see Section~\ref{sec:1.302}). 
This factor is tight for any combination of partitioning using exponential clocks and thresholding schemes with ``analysis based on two thresholds'' (see Appendix~\ref{sec:1.302-tight} for details).

Finally, we show that this factor can be further improved by including another rounding scheme, the scheme of {\em independent thresholds} from \cite{KKSTY04}. However, in this case we are not able to analyze the approximation factor manually anymore. With the help of an LP solver, we find a set of parameters that leads to an approximation factor of $1.2965$ (see Section~\ref{sec:below-1.3}). The verification of this result reduces to finding the maximum of a certain function of $2$ variables (involving polynomials and exponentials), which we have done by computer.\footnote{We have used {\em IBM ILOG CPLEX} for linear programming and {\em Wolfram Mathematica} for analytical manipulations.} 

\smallskip
\noindent{\bf Pairwise realizable distributions.}
While searching for possible extensions of the rounding schemes with a random threshold for each terminal (see Section~\ref{sec:1.302-pairwise} for more details), we encountered the following question:

{\em Given a joint distribution $\rho$ of two random variables $(X,Y)$, can we design arbitrarily many random variables $X_1,\ldots,X_k$ such that $\forall i \neq j$, $(X_i,X_j)$ has the same distribution $\rho$?}

If this is the case, we call such a distribution $\rho$ {\em pairwise realizable}.
Not every distribution $\rho$ is pairwise realizable (see Appendix~\ref{sec:pairwise}).
Here we present the following answer (in discrete domains): $\rho$ is pairwise realizable, if and only if $\rho$ is a convex combination of symmetric product distributions, or in other words 
$$ \Pr_{(X,Y) \sim \rho}[X=a,Y=b] = \sum_s \alpha_s p_s(a) p_s(b) $$
where $\alpha_s \geq 0, \sum_s \alpha_s = 1$ and $\forall s; p_s(a) \geq 0, \sum_a p_s(a) = 1$. We provide a short proof in Appendix~\ref{sec:pairwise}; this result also follows from \cite{TW98}. In particular, it is necessary that the matrix $P_{ab} = \Pr[X=a, Y=b]$ be {\em positive semidefinite}. We do not need this result directly for any of our algorithms, but this characterization is helpful in understanding what kinds of threshold distributions are worth considering. 
(See Section~\ref{sec:1.302-pairwise} and Appendix~\ref{sec:pairwise} for more details.) 

\noindent{\bf Discussion.}
We have investigated several rounding schemes that might improve the approximation factor.
While we have a good understanding of thresholding schemes that rely only on 2 relevant variables in the analysis
 (and we identify the best approximation factor in this setting - see Section~\ref{sec:1.302-tight}),
 the situation gets more complicated with the inclusion of additional variables as in the analysis of {\em independent thresholds}. Then the cut density is not a linear function of the underlying probability distributions anymore. Our approach in this case is a combination of intuition from tight examples and the use of an LP solver.
The rounding scheme achieving a factor of $1.2965$ that we present in this paper is the best one that we are able to describe in a simple form. Further (small) improvements might be achieved by finding more exhaustive descriptions of the probability distributions returned by the LP solver. However, we do not think that this would improve our understanding of the Multiway Cut problem. 

It is interesting to note that as of now, all known approximation algorithms for Multiway Cut can be implemented using a sequence of label assignments based on a threshold condition ($x_{v,i} \geq \theta$). The only exception to our knowledge, the exponential clocks scheme of \cite{BNS13}, can be replaced by the Kleinberg-Tardos algorithm, which uses a sequence of thresholds (with repeated use of variables). In fact \cite{KKSTY04}, citing computational experiments, speculated that the optimal approximation for Multiway Cut might be achievable using ``sparcs'', which are sequences of $k$ threshold cuts, one for each variable. Our scheme of descending thresholds is of this type. However, the exponential clocks scheme as well as the Kleinberg-Tardos scheme, one of which is still a necessary ingredient in our algorithm, are outside of this framework. 

\section{The CKR Relaxation}
\citet{CKR01} proposed the following LP relaxation of the Multiway Cut problem, where $V$ is the set of vertices, $E$ is the set of edges with weights $w_{v,v'}$, and $T$ denotes the set of terminals, $|T|=k$.

\begin{align}
\min~ \frac{1}{2} \sum_{(v,v') \in E} w_{v,v'} \|\bx_{v} - \bx_{v'}\|_{1}: \\
\forall v \in V, ~ \|\bx_{v}\|_{1} = 1,\\
\forall t \in T, ~ \bx_{t} = \b1_{t},\\
\forall v \in V, ~ \bx_{v} \ge {\bf 0}~.
\end{align}



Here, $\bx_v \in \RR^k$ for each vertex $v \in V$, and $\b1_t$ denotes the unit basis vector corresponding to terminal $t \in [k]$. 
The fractional solution can be viewed as an embedding of the graph in the unit simplex $\Delta = \{ \bx \in \RR^k: \bx \geq 0, \|\bx\|_1 = 1 \}$, with terminal at the vertices of $\Delta$.
Given a fractional solution, the objective of any rounding scheme is to assign each of the vertices to one of the terminals without increasing the objective value $\frac{1}{2}\sum_{(v,v') \in E} w_{v,v'} \|\bx_{v} - \bx_{v'}\|_{1}$ by much.
For a rounding scheme to achieve an $\alpha$ approximation to the LP ($\alpha \ge 1$), it suffices to show that for every edge $(v,v') \in E$, the probability that the edge is ``cut'' by the rounding scheme (that is $v$ and $v'$ are assigned to different terminals) is at most $\alpha \cdot \frac{1}{2} \|\bx_{v} - \bx_{v'}\|_{1}$. We call $\frac12 \|\bx_{v} - \bx_{v'}\|_1$ the {\em length} of the edge $(v,v')$.

Moreover, as has been shown by \citet{CKR01}, it suffices to consider the case where for each edge $(v,v') \in E$, the two end-points $v$ and $v'$ are mapped so that their corresponding vectors differ in only two coordinates. In other words, we can assume that $\bx_{v} = (u_{1}, u_{2}, \cdots, u_{k})$ and $\bx_{v'} = (u_{1}, u_{2}, \cdots, u_{i} + \epsilon, \cdots, u_{j} - \epsilon, \cdots, u_{k})$, where $i$ and $j$ are the two coordinates where the two vectors differ. Also, note that in this case, $\frac{1}{2}\|\bx_{v} - \bx_{v'}\|_{1} = \epsilon$.
The probability of cutting such an edge should be at most $\alpha \epsilon$. In fact, $\epsilon$ can be made arbitrarily small, by subdividing edges. Dividing the cut probability by the length of the edge and letting $\epsilon \rightarrow 0$, we obtain the notion of {\em cut density}.
\begin{definition}
A randomized rounding scheme is a probability distribution $\cal R$ over labelings $\ell:\Delta \rightarrow [k]$.
An edge of type $(i,j)$ is an edge $(v,v')$ where $\bx_v$ and $\bx_{v'}$ differ only in coordinates $i,j$.
For a randomized rounding scheme $\cal R$, the cut density for edges of type $(i,j)$ at $\bx \in \Delta$ is
\begin{align*}d_{ij}(\bx) = \limsup_{\epsilon \rightarrow 0} \frac{\Pr_{\ell \sim {\cal R}}[\ell(\bx) \neq \ell(\bx+\epsilon \b1_i - \epsilon \b1_j)]}{\epsilon}.\end{align*}
\end{definition}

As shown in \citet{KKSTY04}, to achieve an approximation factor of $\alpha$ for Multiway Cut it is sufficient to demonstrate a rounding scheme such that $d_{ij}(\bx) \leq \alpha$ for all $i,j \in [k]$ and $\bx \in \Delta$.





\section{Exponential Clocks \& Single Threshold: $1.309017$-approximation}
\label{sec:1.309}

We begin with a rounding scheme based on the techniques of \citet{BNS13} that achieves a $(3+\sqrt{5})/4 \approx 1.309017$-approximation for the Multiway Cut problem. \cite{BNS13} use a combination of two rounding schemes, the ``exponential clocks rounding scheme'' and what we call the ``single-threshold rounding scheme''.  The two schemes are described in detail below. 

\begin{algorithm}[H]
\caption{}
\label{alg:1.309}
\begin{algorithmic}
\STATE With probability $p$, choose the Exponential Clocks Rounding Scheme (Algorithm~\ref{alg:exp-clock}).
\STATE With probability $1-p$, choose the Single Threshold Rounding Scheme (Algorithm~\ref{alg:single-threshold}).
\end{algorithmic}
\end{algorithm}
\begin{algorithm}[H]
\caption{Exponential Clocks Rounding Scheme}
\label{alg:exp-clock}
\begin{algorithmic}[1]
\STATE Choose independent random variables $Z_{i}$ from the exponential distribution for $i=1,\cdots,k$.
\STATE For each vertex $v$, assign $v$ to $\arg\min_{i \in [k]} Z_{i}/x_{v, i}$.
\end{algorithmic}
\end{algorithm}

\begin{algorithm}[H]
\caption{Single Threshold Rounding Scheme}
\label{alg:single-threshold}
\begin{algorithmic}[1]
\STATE Choose a threshold $\theta \in (0,1]$ with probability density $\phi(\theta)$.
\STATE Choose a random permutation $\sigma$ of the terminals.
\FORALL{$i$ in $[k-1]$}
\STATE For every vertex $v=(v_{1}, \cdots, v_{k})$ such that $v$ has not been assigned yet, 
\STATE assign $v$ to terminal $\sigma(i)$ if $x_{v,\sigma(i)} \ge \theta$.
\ENDFOR
\STATE Assign all remaining unassigned vertices to terminal $\sigma(k)$.
\end{algorithmic}
\end{algorithm}

Algorithm~\ref{alg:exp-clock} can be viewed as selecting a point $\bz$ uniformly in the simplex $\Delta$ (by taking  $\bz = \frac{(Z_1,Z_2,\ldots,Z_k)}{Z_1+Z_2+\ldots+Z_k}$) and then partitioning the simplex into $k$ regions that meet at the point $\bz$. \cite{BNS13} also observe that a rounding scheme from \cite{KT02} can be used in place of Algorithm~\ref{alg:exp-clock} and leads to the same cut density of edges. In Algorithm~\ref{alg:single-threshold}, a threshold $\theta$ is chosen from some distribution and the simplex is partitioned by going over a random permutation of the terminals and assigning all currently unassigned points $\bx$ such that $x_{\sigma(i)} \ge \theta$ to terminal $\sigma(i)$.

Let us recall the bounds on cut density for Algorithms~\ref{alg:exp-clock} and \ref{alg:single-threshold} from \cite{BNS13}. By symmetry, we focus in the following on edges of type $(1,2)$.

\begin{lemma}[\cite{BNS13}]
\label{lem:density}
The cut density for edges of type $(1,2)$ under Algorithm~\ref{alg:exp-clock} is
$$ d_{12}(u_1,u_2,\ldots) = 2 - u_1 - u_2.$$
The cut density for edges of type $(1,2)$ under Algorithm~\ref{alg:single-threshold}, assuming $u_1 \leq u_2$, is
$$ d_{12}(u_1,u_2,\ldots) = \frac12 \phi(u_1) + \phi(u_2).$$
\end{lemma}

In order to balance Algorithm~\ref{alg:exp-clock} and Algorithm~\ref{alg:single-threshold}, \cite{BNS13} define the probability distribution $\phi(u)$ as a certain power of $u$. Our rounding scheme differs in the choice of this probability distribution. (In fact we claim that we have identified the best possible distribution for this purpose --- see Section~\ref{sec:tight-example}). We prove the following.

\begin{theorem}
\label{thm:1.309}
Algorithm~\ref{alg:1.309}, with $p = \frac{5 + 3\sqrt{5}}{20} \approx 0.58541$ and probability density function (for Algorithm~\ref{alg:single-threshold})
\begin{equation}
\phi(u) = \left\{
\begin{array}{rl}
a~u & \mbox{for } 0\le u\le b\\
\frac{a}{2}~(u + b) & \mbox{for } b \le u \le 1
\end{array}
\right.
\end{equation}
where $a = \frac{4 + 2\sqrt{5}}{3}$ and $b = \sqrt{5}-2$, achieves a $\frac{3+\sqrt{5}}{4} \approx 1.309017$-approximation for Multiway Cut.
\end{theorem}

The intuition behind this construction is as follows. Assume that $u_1 < u_2$. Considering Lemma~\ref{lem:density}, we would ideally like to design $\phi$ so that $\frac12 \phi(u_1) + \phi(u_2) = c(u_1+u_2) + d$ for some constants $c,d$, in order to combine it with the cut density of $2-u_1-u_2$ for Algorithm~\ref{alg:exp-clock}. Unfortunately, it is impossible to design $\phi$ in such a way: If such a function $\phi$ existed, it would have to satisfy $\frac12 \phi'(u_1) = \phi'(u_2)$ for every pair of values $u_1 < u_2$ which is not possible.
Still, we can try to satisfy this property for many pairs, and our construction (Figure~\ref{fig:1.309-phi}) achieves this for all pairs such that $u_1 < b < u_2$. This means that our analysis is going to be tight for all such $(u_1,u_2)$ pairs.

\begin{center}
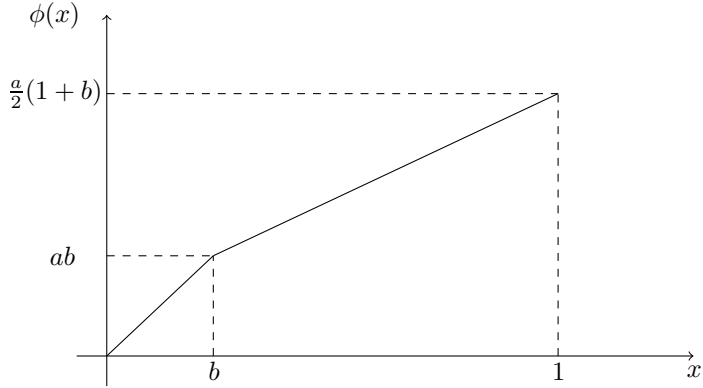
\begin{figure}[h]
\pgfmathsetmacro{\Scale}{2}
\pgfmathsetmacro{\xScale}{3}
\pgfmathsetmacro{\leftShiftLow}{1.4}
\pgfmathsetmacro{\leftShiftHigh}{1.6}


\pgfmathsetmacro{\PointB}{(sqrt(5)-2)*\xScale}
\pgfmathsetmacro{\PointOne}{\xScale}
\pgfmathsetmacro{\PointA}{(4+2*sqrt(5))/3}
\pgfmathsetmacro{\PointTopP}{\PointA*\PointB/\xScale}
\pgfmathsetmacro{\PointTopQ}{\PointA*(\PointOne+\PointB)/(2*\xScale)}

\begin{tikzpicture}[scale=\Scale]

\hspace{80pt}

\node (a1) at (\PointB - \leftShiftLow,-0.1) {$b$};
\node (a1) at (\xScale - \leftShiftLow,-0.1){1};
\node (a1) at (-0.1- \leftShiftHigh,\PointTopP){$ab$};
\node (a1) at (-0.15- \leftShiftHigh,\PointTopQ){$\frac{a}{2} (1+b)$};
\node (a1) at (1.3*\xScale - \leftShiftLow, -0.1){$x$};
\node (a1) at (-0.15- \leftShiftHigh, 1.3*\PointTopQ) {$\phi(x)$};

\draw[->] (-0.2,0) -- (1.3*\xScale,0);
\draw[->] (0,-0.2) -- (0,1.3*\PointTopQ) ;

\draw[-] (0,0) -- (\PointB,\PointTopP);
\draw[-] (\PointB,\PointTopP) -- (1*\xScale, \PointTopQ);

\draw[dashed, -] (\PointB,0) -- (\PointB, \PointTopP);
\draw[dashed, -] (0, \PointTopP) -- (\PointB, \PointTopP);

\draw[dashed, -] (\xScale, 0) -- (\xScale, \PointTopQ);
\draw[dashed, -] (0, \PointTopQ) -- (\xScale, \PointTopQ);

\end{tikzpicture}

\caption{The probability density function $\phi(x)$.}
\label{fig:1.309-phi}
\end{figure}
\end{center}


Now we calculate the optimal values of $a$ and $b$. A constraint on $a$ and $b$ is that since $\int_{0}^{1}\phi(u) \mathrm{d}u=1$, we should have $\frac12 a b^2 + \frac{a}{2} (\frac12(1+b) + b)(1-b) = \frac14 a(-b^{2}+2b+1)=1$ ($a\ge 0$ and $b\in [0,1]$). 
The following lemma gives the total cut density under Algorithm~\ref{alg:1.309}.

\begin{lemma}
The cut density for an edge of type $(1,2)$ at $(u_1,u_2,\ldots)$, $u_1 \leq u_2$ under Algorithm~\ref{alg:1.309} is at most
$$ p \cdot (2 - u_{1} - u_{2}) + (1-p) \cdot \frac{a}{2} (u_{1}+ u_{2} + b).$$
\end{lemma}

\begin{proof}
Observe that for any $a, b \geq 0$, the density function $\phi(u)$ can be written equivalently as $\phi(u) = \min \{ au, \frac{a}{2} (u+b) \}$.
Plugging this expression into Lemma~\ref{lem:density}, the cut density under Algorithm~\ref{alg:single-threshold} is at most
$$ \frac12 \phi(u_1) + \phi(u_2) \leq \frac12 a u_1 + \frac{a}{2} (u_2 + b) = \frac{a}{2} (u_1 + u_2 + b).$$
Since we take Algorithm~\ref{alg:exp-clock} with probability $p$ and Algorithm~\ref{alg:single-threshold} with probability $1-p$, the lemma follows.
\end{proof}

Hence, we can upper-bound the cut density under Algorithm~\ref{alg:1.309} by $ p (2 - u_{1} - u_{2}) + (1-p) \frac{a}{2} (u_{1}+ u_{2} + b)$. To eliminate the dependence on $u_1+u_2$, we set $p = (1-p)\frac{a}{2}$, which means $p = a / (2+a)$. This makes the bound equal to  $2p + (1-p) \frac{a}{2} b = 2p + p b = (2+b)a/(2+a)$.
Hence, the final expression that we would like to minimize is $(2+b)a/(2+a)$ subject to the constraint $\frac14 a(-b^{2}+2b+1)=1$. The minimum is achieved at $a = \frac23(2+\sqrt{5})$ and $b=\sqrt{5}-2$, where the bound on cut density is $(2+b) a / (2+a) = \frac14 (3+\sqrt{5})$.
The probability $p$ in Algorithm~\ref{alg:1.309} is $p =a/(2+a)=\frac{1}{20}(5+3\sqrt{5})$. This proves Theorem~\ref{thm:1.309}.
\section{Descending Thresholds: $1.30217$-approximation}
\label{sec:1.302}

As we show in Appendix~\ref{sec:tight-example}, the Exponential Clocks Rounding Scheme combined with the Single Threshold Rounding Scheme (under any threshold distribution) achieves exactly the factor of $\frac{3+\sqrt{5}}{4}$ and not better.
In this section, we present an improved $\frac{10+4\sqrt{3}}{13} \approx 1.30217$-approximation for Multiway Cut. This is achieved by combining the techniques of \cite{BNS13} with a new rounding scheme that we call {\em descending thresholds}. Before we describe our algorithm, let us discuss the ideas and considerations that led us to this rounding scheme.

\subsection{Pairwise realizable distributions}
\label{sec:1.302-pairwise}

The tight example in Appendix~\ref{sec:tight-example} serves as a test of scrutiny for any candidate rounding technique: If it does not provide a factor better than $\frac{3+\sqrt{5}}{4}$ on this example, then it will not be helpful in improving the approximation factor. In particular, we know from Appendix~\ref{sec:tight-example} that if we want to use single-coordinate cuts in the form $\{i: x_{ij} \geq \theta\}$, we cannot use the same threshold $\theta$ for all terminals.

It is natural to consider different thresholds for different terminals, but the space of possibilities (all joint probability distributions of $(\theta_1, \theta_2, \ldots, \theta_k)$) seems too vast to explore. However, let us make the following observation: for an edge of type $(i,j)$, only the thresholds $\theta_i, \theta_j$ corresponding to terminals $i,j$ can potentially cut this edge. Therefore, the cut density for edges of type $(i,j)$ is primarily determined by the distribution of the thresholds $\theta_i, \theta_j$. Assuming that threshold $\theta_i$ is applied before threshold $\theta_j$ and the joint distribution of $\theta_i, \theta_j$ is given by a density function $\pi(\theta_i, \theta_j)$, we can write the following bound on the cut density of an edge of type $i,j$ located at $\bu$:
$$ \int_0^1 \pi(u_i, u) du + \int_{u_i}^{1} \pi(u, u_j) du.$$
Note the asymmetry here: the edge is always cut in coordinate $u_i$ if $\theta_i$ cuts it, but it is cut in coordinate $u_j$ only if $\theta_j$ cuts it and the edge was not captured by terminal $i$ before.

Eventually, we apply a sequence of cuts to a permutation of terminals, either an independently random one, or one correlated with the values of the thresholds in some way.  Given such a rounding scheme, we can compute a bound on the cut density as above for edges of all types. Since the bound is linear as a function of $\pi$, one can formulate a linear program that searches for the best distribution $\pi$. As shown in \cite{KKSTY04}, the symmetry of the terminals implies that whatever rounding scheme we have, we can apply it after a random re-labeling of the terminals. Therefore, we can assume that overall, every pair of thresholds $(\theta_i, \theta_j)$ has the same joint distribution, which is a symmetrization of the distribution $\pi$ above: $\rho(\theta_i,\theta_j) = \frac12 (\pi(\theta_i,\theta_j) + \pi(\theta_j,
\theta_i))$. A basic question that we encountered here is: What distributions $\rho$ are actually realizable for all pairs of terminals at the same time? 
%
In Appendix~\ref{sec:pairwise}, we provide the following characterization, at least in a discrete setting:
A distribution $\rho$ can be realized for all pairs of terminals simultaneously,
if and only if $\rho$ is a convex combination of symmetric product distributions, in other words $\rho$ is in the form 
$$ P_{ab} = \Pr[X=a,Y=b] = \sum_s \alpha_s p_s(a) p_s(b) $$
where $\alpha_s, p_s(a) \geq 0$ and $\sum_s \alpha_s = 1, \sum_a p_s(a) = 1$
(see Theorem~\ref{thm:pairwise-realizable}). In particular, this implies that the matrix $P_{ab} = \Pr[X=a, Y=b]$ must be {\em positive semidefinite}.

Going back to our original motivation, this characterization gives us a hint as to what kinds of distributions over thresholds are worth considering. It is sufficient to consider a combination of rounding schemes, where the thresholds in each scheme are chosen independently from a certain distribution, and then applied in a certain order (which is possibly correlated with their values;
 {\em this is an aspect independent of the notion of pairwise realizability}).
Unfortunately, we do not know how to search efficiently over the space of all such distributions. As far as we know, the condition of being pairwise-realizable is not equivalent to $P_{ab}$ being positive semidefinite. However, when we computed the best distribution $\pi^*$ to be combined with the Exponential Clocks Rounding Scheme, without the restriction of being pairwise-realizable, the symmetrization of this distribution $\rho^* = \frac12 (\pi^* + {\pi^*}^T)$ just happened to be in the pairwise realizable form. Furthermore, we identified this optimal solution $\pi^*$ as a convex combination of two rounding schemes, the Single Threshold Rounding Scheme under a certain distribution $\phi$, and a ``Descending Thresholds Rounding Scheme'' under a certain distribution $\psi$ --- we describe this scheme in the following section.

\subsection{Descending thresholds}



Based on the discussion above and our computational experiments, we propose the following new rounding scheme.
\begin{algorithm}[H]
\caption{Descending Thresholds Rounding Scheme}
\label{alg:descending}
\begin{algorithmic}[1]
\STATE For each $i \in [k]$, choose independently a threshold $\theta_i \in (0,1]$ with probability density $\psi(\theta)$.
\STATE Let $\sigma$ be a permutation of $[k]$ such that $\theta_{\sigma(1)} \geq \theta_{\sigma(2)} \geq \ldots \geq \theta_{\sigma(k)}$.
\FORALL{$i$ in $[k-1]$}
\STATE For every vertex $v \in V$ that has not been assigned yet, 
\STATE if $x_{v,\sigma(i)} \ge \theta_{\sigma(i)}$ then assign $v$ to terminal $\sigma(i)$.
\ENDFOR
\STATE Assign all remaining unassigned vertices to terminal $\sigma(k)$.
\end{algorithmic}
\end{algorithm}

This scheme is in fact quite easy to analyze. The cut density for an edge of given location is given by the following lemma.

\begin{lemma}
\label{lem:desc-cut-density}
The cut density under Algorithm~\ref{alg:descending} for edges of type $(1,2)$ at $(u_1,u_2,\ldots,u_k)$ such that $u_1 \leq u_2$ is at most
$$ \left(1 - \int_{u_1}^{u_2} \psi(u) du \right) \psi(u_1) + \psi(u_2).$$
\end{lemma}

\begin{proof}
The edge can be cut in coordinate $u_1$ only if $\theta_2$ is not between $u_1$ and $u_2$: if $u_1 < \theta_2 < u_2$, threshold $\theta_2$ would be considered before $\theta_1$ and the entire edge would be assigned to terminal $2$. This accounts for the first term. The edge can be also cut in coordinate $u_2$, which accounts for the second term. We neglect the possibility of the entire edge being assigned to another terminal, which can only decrease the probability of being cut.
\end{proof}

This lemma provides another way to explain why descending thresholds might be advantageous. The Exponential Clocks Rounding Scheme cuts edges with probability density $2 - u_1 - u_2$, so our goal in the remaining schemes is to achieve a cut density bounded by a linear function of $u_1 + u_2$. The Single Threshold Rounding Scheme achieves this with some slack, because of the asymmetry in its analysis: it cuts edges with $u_1 = u_2$ less often than edges where $u_1$ and $u_2$ are far apart. The Descending Thresholds Rounding Scheme alleviates this problem, because its cut density is lower for edges where $u_1$ and $u_2$ are far apart due to the $\int_{u_1}^{u_2} \psi(u) du$ term.

\subsection{The $1.30217$-approximate rounding scheme}

It remains to describe how we combine the Descending Thresholds Rounding Scheme with the Exponential Clocks and the Single Threshold Rounding Scheme, in particular how we choose the probability distributions $\phi$ and $\psi$. Based on computational evidence, we determined that $\phi$ is piecewise linear with a breakpoint at $b \in (0,1)$, and the cut density due to the Single Threshold Rounding Scheme is $\frac12 \phi(u_1) + \phi(u_2)$, piecewise linear in $u_1$ and $u_2$. Ideally we want a function in the form $a(u_1 + u_2)$, but as we saw in Section~\ref{sec:1.309} that is hard to achieve. This is where the scheme of Descending Thresholds comes in: We choose the probability distribution $\psi$ for decreasing thresholds uniform in the interval $[0,b]$, so that the cut density due to this scheme, $(1 - \int_{u_1}^{u_2} \psi(u) du) \psi(u_1) + \psi(u_2)$, is again piece\-wise linear in $u_1$ and $u_2$. Then we are able to balance the parameters so that the total cut density is a constant throughout most of the simplex. Our algorithm will be in the following form.
\begin{center}
\vspace{-0.15in}
\begin{algorithm}[H]
\caption{}
\label{alg:1.302}
\begin{algorithmic}
\STATE
\begin{itemize}
\item With probability $p_{1}$, choose the Exponential Clocks Rounding Scheme (Algorithm~\ref{alg:exp-clock}).
\item With probability $p_{2}$, choose the Single Threshold Rounding Scheme (Algorithm~\ref{alg:single-threshold}), where the threshold is chosen with the following probability density:
\begin{compactitem}
\item For $0 \leq u \leq b$, $\phi(u) = a~u$
\item For $b < u \leq 1$, $\phi(u) = c~u+d$
\end{compactitem}
\item With probability $p_{3}$, choose the Descending Thresholds Rounding Scheme (Algorithm~\ref{alg:descending}), where the thresholds are chosen with the following probability density:
\begin{compactitem}
\item For $0 \leq u \leq b$, $\psi(u) = \frac{1}{b}$
\item For $b < u \leq 1$, $\psi(u) = 0$
\end{compactitem}
\end{itemize}
\end{algorithmic}
\end{algorithm}
\end{center}
\vspace{-0.1in}
Using Lemma~\ref{lem:density} and Lemma~\ref{lem:desc-cut-density},
the cut density for edges of type $(1,2)$ located at $(u_1,u_2,\ldots)$ is equal to $q = q_1 + q_2 + q_3$, where:

\noindent
{\bf Case 1:} $0\le u_{1}\le u_{2} \le b$.
\begin{compactitem}
\item Exponential Clocks: $q_1 = p_{1}~(2-u_{1}-u_{2})$.
\item Single Threshold: $q_2 = p_{2} ~ (\frac{a}{2}~u_{1} + a~u_{2})$.
\item Descending Thresholds: $q_3 = p_{3}~((1- \frac{u_{2}-u_{1}}{b})\frac{1}{b} + \frac{1}{b})$.
\end{compactitem}
\vspace{-0.05in}
 \begin{align*}q = 2p_{1} + p_{3} \frac{2}{b}  + \left(-p_{1} + p_{2} \frac{a}{2} + p_{3} \frac{1}{b^{2}}\right) u_{1} + \left(-p_{1} + p_{2} a -p_{3}\frac{1}{b^{2}}\right) u_{2}~.\end{align*}
\noindent
{\bf Case 2:} $0\le u_{1} \le b < u_{2} \le 1$.
\begin{compactitem}
\item Exponential Clocks: $q_1 = p_{1}~(2-u_{1}-u_{2})$.
\item Single Threshold: $q_2 = p_{2}~(\frac{a}{2}~u_{1} + c~u_{2}+d)$.
\item Descending Thresholds: $q_3 = p_{3}~\frac{u_{1}}{b^{2}}$.
\end{compactitem}
$$q = 2p_{1} + p_{2}d + \left(-p_{1} + p_{2} \frac{a}{2} + p_{3}\frac{1}{b^{2}}\right) u_{1} + \left(-p_{1} + p_{2}c \right) u_{2}~.$$

\noindent
{\bf Case 3:} $b < u_{1} \le u_{2} \le 1$.
\begin{compactitem}
\item Exponential Clocks: $q_1 = p_{1}~(2-u_{1}-u_{2})$.
\item Single Threshold: $q_2 = p_{2}~(\frac{c}{2}u_{1} + c~u_{2} + \frac{3}{2}d)$.
\item Descending Thresholds: $q_3 = 0$.
\end{compactitem}
$$ q = 2p_{1} + \frac{3}{2}p_{2}d + \left(-p_{1} + p_{2} \frac{c}{2} \right) u_{1} + \left(-p_{1} + p_{2} c \right) u_{2}~.$$

\noindent
{\bf Optimizing the parameters.}
Let us denote $\tilde{a} = p_{2}a$, $\tilde{c} = p_{2}c$ and $\tilde{d} = p_{2}d$. Given the cut density formulas above, we would like to minimize $z$ subject to the constraints

\begin{eqnarray}
\forall~0\le u_{1}\le u_{2} \le b; & & z \ge 2p_{1} + \frac{2}{b} p_{3} + (-p_{1} + \frac12 \tilde{a} + \frac{1}{b^{2}} p_{3}) u_{1} + (-p_{1} + \tilde{a} - \frac{1}{b^2} p_{3}) u_{2} \label{eq:zeroToB} \\
 \forall~0\le u_{1} \le b \le u_{2} \le 1; & & z \ge 2p_{1} + \tilde{d} + (-p_{1} + \frac12 \tilde{a} + \frac{1}{b^{2}} p_{3}) u_{1} + (-p_{1} + \tilde{c}) u_{2} \label{eq:zeroToBToOne} \\
\forall~b\le u_{1} \le u_{2} \le 1; & & z \ge 2p_{1} + \frac{3}{2} \tilde{d} + (-p_{1} + \frac12 \tilde{c}) u_{1} + (-p_{1} + \tilde{c}) u_{2}\label{eq:bToOne} \\
& & p_{1} + \frac12 \tilde{a}b^{2} + \frac12 \tilde{c}(1-b^{2}) + (1-b) \tilde{d} + p_{3} = 1\label{eq:sum-to-one}\\
& & 0 \le b \le 1 \\
& & p_{1},p_{3}, \tilde{a}, \tilde{c}, \tilde{d} \ge 0
\end{eqnarray}

Equation~(\ref{eq:sum-to-one}) follows from combining the probability normalization conditions $p_{1} + p_{2} + p_{3} = 1$ and $\int_0^1 \phi(u) du = \frac12 ab^{2} + \frac12 c(1-b^{2}) + (1-b)d=1$. In order to eliminate the variables $u_1, u_2$, we impose the following conditions:
 $p_1 = \tilde{c}$,
 $p_{3} = \frac{b^{2}}{3}\tilde{c}$,
 $\tilde{a} = \frac{4}{3}\tilde{c}$,
 and $\tilde{d} = \frac{2}{3}b\tilde{c}$.
This replaces all the constraints on $z$ by $z \geq (2 + \frac{2}{3} b) \tilde{c}$. (In constraint~(\ref{eq:bToOne}), the right-hand side becomes $2 \tilde{c} + b \tilde{c} - \frac12 \tilde{c} u_1$ which is dominated by $2 \tilde{c} + \frac12 b \tilde{c}$ for $u_1 \geq b$.) Constraint (\ref{eq:sum-to-one}) becomes $\tilde{c} (\frac32 + \frac23 b - \frac16 b^2) = 1$.

Therefore, we want to minimize $(2 + \frac23 b) \tilde{c}$ subject to $\tilde{c} (\frac32 + \frac23 b - \frac16 b^2) = 1$, which can be done by hand.  The optimal solution is $b = 2 \sqrt{3} - 3 ~(\approx 0.464102)$ and $\tilde{c} = (6 + 5 \sqrt{3})/26 ~(\approx 0.563856)$, which gives
$z = (10+4\sqrt{3})/13 \leq 1.30217$. The value of the other parameters can be derived from the equations above:
 $p_{1} = (6+5\sqrt{3})/26~(\approx 0.563856)$, $p_2 = (19-8\sqrt{3})/13~(\approx 0.395661)$, $p_{3} = (11\sqrt{3} - 18)/26~(\approx 0.040483)$, $\tilde{a} = (12+10\sqrt{3})/39~(\approx 0.75181)$ and $\tilde{d} = (4-\sqrt{3})/13~(\approx 0.17446)$.
It can be verified that the cut density in all cases is bounded by $z = (10+4\sqrt{3})/13$. We summarize in the following theorem.

\begin{theorem}
\label{thm:1.302}
Algorithm~\ref{alg:1.302} with parameters $p_1 = (6+5\sqrt{3})/26$, $p_2 = (19-8\sqrt{3})/13$, $p_3 = (11\sqrt{3}-18)/13$, $\tilde{a} = (12+10\sqrt{3})/39$, $b = 2\sqrt{3}-3$, $\tilde{c} = (6+5\sqrt{3})/26$ and $\tilde{d} = (4-\sqrt{3})/13$
gives a $(10 + 4\sqrt{3})/13 \simeq 1.30217$-approximation for the Multiway Cut problem.
\end{theorem}

\section{Independent Thresholds: Getting Below $1.3$}
\label{sec:below-1.3}

As we show in  Appendix~\ref{sec:1.302-tight}, the approximation factor of $\frac{10+4\sqrt{3}}{13} \approx 1.30217$ is the best one we can achieve if we combine Exponential Clocks with any sequence of $k$ single-variable cuts, as long as the analysis for edges of type $(i,j)$ works only with variables $u_i,u_j$. An obvious question is what happens if we consider the role of coordinates other than $u_i, u_j$. In particular, an edge of type $(i,j)$ can be ``captured" by another terminal $\ell$, in the sense that both of its endpoints are labeled $\ell$ before coordinates $u_i, u_j$ are even considered. \cite{KKSTY04} rely on this argument to improve their analysis; they use a rounding scheme in the following form.

\begin{algorithm}[H]
\caption{Independent Thresholds Rounding Scheme}
\label{alg:independent}
\begin{algorithmic}[1]
\STATE For each $i \in [k]$, choose independently a threshold $\theta_i \in (0,1]$ with probability density $\xi(\theta)$.
\STATE Let $\sigma$ be a uniformly random permutation of $[k]$.
\FORALL{$i$ in $[k-1]$}
\STATE For every vertex $v \in V$ that has not been assigned yet, 
\STATE if $x_{v,\sigma(i)} \ge \theta_{\sigma(i)}$ then assign $v$ to terminal $\sigma(i)$.
\ENDFOR
\STATE Assign all remaining unassigned vertices to terminal $\sigma(k)$.
\end{algorithmic}
\end{algorithm}

In this section, we pursue further improvements to the approximation ratio, using the Independent Thresholds Rounding Scheme. Unfortunately, the inclusion of further coordinates in the analysis leads to more involved expressions which are non-linear in the underlying distributions. In contrast to the results above, we are no longer able to find the best parameters and verify the approximation ratio by hand. Our algorithm in this section has been found with the help of an LP solver; we provide more details on our computational experiments below. Let us remark that the LP solution involves a discretized description of a probability distribution $\phi$, which makes it difficult to even present a description of the algorithm in a concise form. However, we have been able to find an (approximate) closed-form expression for $\phi$ and describe the algorithm in a compact form (see below). This incurs a small loss in the approximation factor. In this section, we do not claim that our approximation is optimal for any particular class of rounding schemes; our aim is to demonstrate that the approximation factor can be pushed below $1.3$.

{\em Why the gains are small.} Before we get into the details of our final algorithm and its analysis, let us comment on why it seems difficult to achieve very substantial gains at this point. The gain from considering coordinates other than $u_i, u_j$ comes from the fact that a threshold $\theta_\ell$ might capture an edge of type $(i,j)$ if the respective coordinate $u_\ell$ is above $\theta_\ell$. However, given $u_i$ and $u_j$, all the other coordinates could be very small, namely $u_\ell = (1-u_i-u_j) / (k-2)$. For the rounding schemes that we considered above (Single Threshold and Descending Thresholds), this means that we do not get any improvement, because for any $u_i,u_j>0$, we can have $u_\ell < u_i, u_j$ for all $\ell \neq i,j$; then, it never happens that terminal $\ell$ captures the edge before $i$ or $j$. In order to exploit this argument, we need to use another rounding scheme, such as Independent Thresholds. (Other schemes are possible but computational experiments show that the Independent Thresholds Rounding Scheme is the most effective one out of those we were able to analyze.)  Moreover, the probability density $\xi(u)$ should be relatively high for $u \rightarrow 0$, to make the probability $\Pr[\theta_\ell \leq u_\ell]$ significant. On the other hand, this is somewhat contrary to the goal of balancing cut densities with the Exponential Clocks Rounding Scheme where the cut density is $2-u_i-u_j$,~i.e.~maximized for edges with $u_i, u_j$ close to $0$. In other words, we cannot use Independent Thresholds with a very high density near $0$ because this would have significant impact on the cut density of edges with $u_i,u_j$ close to $0$. This is our interpretation of why the benefit of the Independent Thresholds Rounding Scheme on top of the other techniques is rather limited.

\subsection{Our $1.2965$-approximation algorithm}

Based on computational experiments, we propose a rounding scheme in the following form, with a probability density $\phi$, parameter $b \in [0,1]$ and probabilities $p_1+p_2+p_3+p_4 = 1$ to be determined.
\vspace{-0.1in}
\begin{center}
\begin{algorithm}[H]
\caption{}
\label{alg:1.3}
\begin{algorithmic}
\STATE
\begin{compactitem}
\item With probability $p_{1}$, choose the Exponential Clocks Rounding Scheme (Algorithm~\ref{alg:exp-clock}).
\item With probability $p_{2}$, choose the Single Threshold Rounding Scheme (Algorithm~\ref{alg:single-threshold}), where the threshold is chosen with probability density $\phi$.
\item With probability $p_{3}$, choose the Descending Thresholds Rounding Scheme (Algorithm~\ref{alg:descending}), where the thresholds are chosen uniformly in $[0,b]$.
\item With probability $p_{4}$, choose the Independent Thresholds Rounding Scheme (Algorithm~\ref{alg:independent}), where the thresholds are chosen uniformly in $[0,b]$.
\end{compactitem}
\end{algorithmic}
\end{algorithm}
\end{center}

We defer the proofs of the results in this section to the full version of the paper. For the following result, we appeal to \cite{KKSTY04} for the analysis of the Independent Thresholds Rounding Scheme.

\begin{lemma}
\label{lem:Karger-cut}
Given a point $(u_1,u_2,\ldots,u_k) \in \Delta$ and the parameter $b$ of Algorithm~\ref{alg:independent}, let $a = \frac{1-u_1-u_2}{b}$.
If $a>0$, the cut density for an edge of type $(1,2)$ located at $(u_1,u_2,\ldots,u_k)$ under the Independent Thresholds Rounding Scheme with parameter $b$ is at most
\begin{itemize}
\item $\frac{2(1-e^{-a})}{ab} - \frac{(u_1+u_2)(1 - (1+a) e^{-a})}{a^2 b^2}$, if all the coordinates $u_1,\ldots,u_k$ are in $[0,b]$.
\item $\frac{a+e^{-a}-1}{a^2 b}$, if $u_1 \in [0,b], u_2 \in (b,1]$ and $u_\ell \in [0,b]$ for all $\ell \geq 3$.
\item $\frac{1}{b} - \frac{u_1+u_2}{6b^2}$, if $u_1,u_2 \in [0,b]$ and $u_\ell \in (b,1]$ for some $\ell \geq 3$.
\item $\frac{1}{3b}$, if $u_1 \in [0,b], u_2 \in (b,1]$ and $u_\ell \in (b,1]$ for some $\ell \geq 3$.
\item $0$, if $u_1, u_2 \in (b,1]$.
\end{itemize}
For $a=0$, the cut density is given by the limit of the expressions above as $a \rightarrow 0$.
\end{lemma}

\begin{proof}
For the first and second case, we refer to the proof of Lemma 6.6 in \cite{KKSTY04}. They prove that if $u_\ell \in [0,b]$ for all $\ell \geq 3$, then the maximum cut density is achieved when $u_\ell = (1-u_1-u_2)/(k-2)$ for all $\ell \geq 3$, and it is bounded by the following formula:
\begin{align*}
C_\infty(u_1,u_2) &= [F'(u_1) + F'(u_2)] \times \frac{1-e^{-a}}{a} 
 - [F'(u_1) F(u_2) + F'(u_2) F(u_1)] \times \frac{1 - (1+a) e^{-a}}{a^2}.
\end{align*}
Here, $F$ is the cumulative distribution function for picking the independent thresholds, namely $F(u) = \min \{u/b, 1\}$. 
Hence, we have $F'(u) = 1/b$ for $u \in [0,b]$ and $F'(u) = 0$ for $u \in (b,1]$.

In the case where $u_1, u_2 \in [0,b]$, we obtain
$$ C_\infty(u_1,u_2) = \frac{2}{b} \times \frac{1-e^{-a}}{a} - \frac{u_1 + u_2}{b^2}  \times \frac{1-(1+a) e^{-a}}{a^2}. $$

In the case where $u_1 \in [0,b], u_2 \in (b,1]$, we obtain
$$ C_\infty(u_1,u_2) = \frac{1}{b} \times \frac{1-e^{-a}}{a} - \frac{1}{b} \times \frac{1-(1+a) e^{-a}}{a^2} = \frac{1}{b} \times \frac{a+e^{-a}-1}{a^2}.$$

Let us now consider the case where $u_\ell > b$ for some $\ell \geq 3$, w.l.o.g.~$u_3 > b$. As shown in Lemma 6.4 in \cite{KKSTY04}, the maximum cut density for such edges is achieved when $u_3 = 1-u_1-u_2 > b$, and then it is equal to
\begin{eqnarray*}
C_3(u_1,u_2,1-u_1-u_2)  & = & d_3(u_1,u_2,1-u_1-u_2)  \\
& = &   \frac16 ((F'(u_1) + (1-F(u_1)) F'(u_2))  + (F'(u_2) + (1-F(u_2)) F'(u_1)) \\
& & + (F'(u_1) + 0 \cdot F'(u_2)) + (F'(u_2) + 0 \cdot F'(u_1)) + 0 + 0)
\end{eqnarray*}
from equation (1) in \cite{KKSTY04}. For $u_1, u_2 \in [0,b]$, we have $F'(u_1) = F'(u_2) = 1/b$ and $F(u_1) = u_1/b, F(u_2) = u_2/b$, hence
\begin{align*}
C_3(u_1,u_2,1-u_1-u_2) &= \frac{1}{6} \left( \frac{1}{b} + \left(1 - \frac{u_1}{b} \right) \frac{1}{b}
  + \frac{1}{b} + \left(1 - \frac{u_2}{b} \right) \frac{1}{b} + \frac{1}{b} + \frac{1}{b}\right) 
 = \frac{1}{b} - \frac{u_1+u_2}{6b^2}.
\end{align*}
For $u_1 \in [0,b], u_2 \in (b,1]$ (which can only occur if $b < 1/2$), we get $F(u_1) = u_1/b, F'(u_1) = 1/b, F(u_2) = 1$ and $F'(u_2) = 0$, hence
$$ C_3(u_1,u_2,1-u_1-u_2) = \frac16 \left( \frac{1}{b} + \frac{1}{b} \right) = \frac{1}{3b}.$$

If $u_1,u_2 \in (b,1]$ then the edge is never cut because thresholds are chosen only in $[0,b]$.

Finally, the cut density for $a=0$ follows by continuity of the density function, or can be verified directly from the properties of the rounding scheme.
\end{proof}

We also refine the analysis of the Single Threshold and Descending Thresholds Rounding Schemes, depending on the value of the remaining coordinates.

\begin{lemma}
\label{lem:single-cut}
For an edge of type $(1,2)$ located at \\$(u_1,u_2,\ldots,u_k)$, the cut density under the Single Threshold Rounding Scheme is at most
\begin{itemize}
\item $\frac12 \phi(u_1) + \phi(u_2)$, if $u_\ell \leq u_1 \leq u_2$ for all $\ell \geq 3$.
\item $\frac13 \phi(u_1) + \phi(u_2)$, if $u_1 < u_\ell \leq u_2$ for some $\ell \geq 3$.
\item $\frac13 \phi(u_1) + \frac12 \phi(u_2)$, if $u_1 \leq u_2 < u_\ell$ for some $\ell \geq 3$.
\end{itemize}
\end{lemma}

\begin{proof}
The first case follows from Lemma~\ref{lem:density}. We use similar reasoning to handle the other two cases. If there is a coordinate $u_\ell$, $u_1 < u_\ell \leq u_2$, then one of the terminals $\{2,\ell\}$ is considered before $1$, then the edge cannot be cut in coordinate $u_1$. Hence, we get a contribution of $\phi(u_1)$ only if $1$ appears before both $2$ and $\ell$, which happens with probability $1/3$.

iIf $u_1 \leq u_2 < u_\ell$ for some $\ell \geq 3$, then we use the same reasoning for cutting the edge in coordinate $u_1$. In addition, the edge can be cut in coordinate $u_2$ only if terminal $2$ is considered before $\ell$, which happens with probability $1/2$. Hence we get a contribution of $\phi(u_1)$ with probability $1/3$ and $\phi(u_2)$ with probability $1/2$.
\end{proof}

\begin{lemma}
\label{lem:descending-cut}
For an edge of type $(1,2)$ located at \\$(u_1,u_2,\ldots,u_k)$, the cut density under the Descending Thresholds Rounding Scheme is at most
\begin{itemize}
\item $(1 - \int_{u_1}^{u_2} \psi(u) du) \psi(u_1) + \psi(u_2)$, if $u_\ell \leq u_1 \leq u_2$ for all $\ell \geq 3$.
\item $(1 - \int_{u_1}^{u_2} \psi(u) du)(1 - \int_{u_1}^{u_\ell} \psi(u) du) \psi(u_1) + \psi(u_2)$, if $u_1 < u_\ell \leq u_2$ for some $\ell \geq 3$.
\item $(1 - \int_{u_1}^{u_2} \psi(u) du)(1 - \int_{u_1}^{u_\ell} \psi(u) du) \psi(u_1)$ \\$+ (1 - \int_{u_2}^{u_\ell} \psi(u) du) \psi(u_2)$, if $u_1 \leq u_2 <  u_\ell$ for some $\ell \geq 3$.
\end{itemize}
\end{lemma}

\begin{proof}
The first case follows from Lemma~\ref{lem:desc-cut-density}. In the second case, we have $u_1 < u_\ell \leq u_2$ and hence the edge can be cut only if it is not captured by terminal $2$ or $\ell$ before terminal $1$ is considered. The edge is captured by terminal $2$ exactly when the threshold $\theta_2$ is between $u_1$ and $u_2$, which happens with probability $\int_{u_1}^{u_2} \psi(u) du$. Independently, the edge is captured by terminal $\ell$ when $\theta_\ell$ is between $u_1$ and $u_\ell$, which happens with probability $\int_{u_1}^{u_\ell} \psi(u) du$. Therefore, the probability of cutting in coordinate $u_1$ is $(1 - \int_{u_1}^{u_2} \psi(u) du)(1 - \int_{u_1}^{u_\ell} \psi(u) du) \psi(u_1)$. The probability of cutting in coordinate $u_2$ remains bounded by $\psi(u_2)$.

In the third case, we have $u_1 \leq u_2 < u_\ell$. The probability of cutting in coordinate $u_1$ remains the same. The probability of cutting in coordinate $u_2$ is now multiplied by the probability that the edge is not captured by terminal $\ell$ before terminal $2$, which is $(1 - \int_{u_2}^{u_\ell} \psi(u) du)$.
\end{proof}

Given these lemmas, we formulate the expressions for the total cut density under Algorithm~\ref{alg:1.3}.

\begin{corollary}
\label{cor:1.3-total-density}
Let $a(u_1,u_2) = (1-u_1-u_2) / b$.
The cut density under Algorithm~\ref{alg:1.3} for an edge of type $(1,2)$ located at $\bu = (u_1,u_2,\ldots)$ is $d_{12}(\bu)$ where
\begin{compactitem}
\item {\bf Case I.} If $u_1 \leq u_2 \leq b$ and $u_\ell \leq b$ for all $\ell \geq 3$, 
\begin{align*}
d_{12}(\bu)  \leq &~~  p_1 \left(2 - u_1 - u_2 \right) + p_2 \left(\frac12 \phi(u_1) + \phi(u_2) \right) \\
 & +  p_3 \left(2 - \frac{1}{b} (u_2-u_1) \right) \frac{1}{b} \\
 & +  p_4 \left( \frac{2(1-e^{-a(u_1,u_2)})}{a(u_1,u_2) b}
 - \frac{(u_1+u_2)(1 - (1+a(u_1,u_2)) e^{-a(u_1,u_2)})}{(a(u_1,u_2) b)^2} \right).
\end{align*}
\item {\bf Case II.} If $u_1 \leq u_2 \leq b$ and $u_\ell > b$ for some $\ell \geq 3$, 
\begin{align*}
d_{12}(\bu) \leq &~~ p_1 \left(2 - u_1 - u_2 \right) +  p_2 \left(\frac13 \phi(u_1) + \frac12 \phi(u_2) \right) \\
 & +  p_3 \left( \left(1 - \frac{1}{b}(u_2-u_1) \right) u_1 + u_2 \right) \frac{1}{b^2} \\
 & +  p_4 \left( 1 - \frac{1}{6b} (u_1+u_2) \right) \frac{1}{b}.
\end{align*}
\item {\bf Case III.} If $u_1 \leq b < u_2$ and $u_\ell \leq b$ for all $\ell \geq 3$,
\begin{align*}
d_{12}(\bu)  \leq &~~  p_1 \left(2 - u_1 - u_2 \right) +  p_2 \left(\frac12 \phi(u_1) + \phi(u_2) \right) \\
 & +  p_3 \frac{u_1}{b^2}  +  p_4 \frac{a(u_1,u_2) + e^{-a(u_1,u_2)} - 1}{(a(u_1,u_2))^2 b}.
\end{align*}
\item {\bf Case IV.} If $u_1 \leq b < u_2$ and $u_\ell > b$ for some $\ell \geq 3$,
\begin{align*}
d_{12}(\bu)  \leq &~~ p_1 \left(2 - u_1 - u_2 \right) +  p_2 \left(\frac13 \phi(u_1) + \phi(u_2) \right) \\
 & +  p_3 \frac{u_1^2}{b^3} +  p_4 \frac{1}{3b}.
\end{align*}
\item {\bf Case V.} If $b < u_1 \leq u_2$,
\begin{align*}
d_{12}(\bu) \leq &~~ p_1 \left(2 - u_1 - u_2 \right) +  p_2 \left(\frac12 \phi(u_1) + \phi(u_2) \right).
\end{align*}
\end{compactitem}
For $a(u_1,u_2)=0$, the expressions should be interpreted as limits when $a(u_1,u_2) \rightarrow 0$.
\end{corollary}

We remark that Case IV and Case V are applicable only when $b < 1/2$ (otherwise we cannot have two variables above $b$). These cases were necessary for us to explore different choices of $b$, but eventually we identified $b = 6/11$ as the optimal choice. (Curiously, \cite{KKSTY04} make the same choice, even though they use only two out of the four rounding schemes present here. We do not quite understand this coincidence.) Hence, Case IV and Case V do not arise in our final verification.

\medskip
{\bf Formulation as an LP.}
Given the analysis of Algorithm~\ref{alg:1.3} above, it remains to choose the parameters $b$, $p_1$, $p_2$, $p_3$, $p_4$ and the probability distribution $\phi(u)$ so that $d_{12}(\bu)$ can be upper-bounded by a constant as small as possible. It is easy to see for a fixed $b$, this can be formulated as a linear program (after discretization). The only seemingly non-linear part of the problem is the product $p_2 \phi(u)$, but this can be easily folded into a single variable $\tilde{\phi}(u) = p_2 \phi(u)$. Finally, we have the normalization constraint $p_1 + p_2 + p_3 + p_4 = p_1 + \int_0^1 \tilde{\phi}(u) du + p_3 + p_4 = 1$. We minimize an upper-bound on $d_{12}(\bu)$ over all values $0 \leq u_1 \leq u_2 \leq 1$, under a suitable discretization. 

We obtained a solution which involves a discretized description of a probability distribution $\phi$. In order to obtain a concise description,  we approximated this probability distribution by a piecewise-polynomial density function which we describe here. Our solution is as follows.
\begin{itemize}
\item $b = 6/11$
\item $p_1 = 0.31052$
\item $p_2 = 0.305782$
\item $p_3 = 0.015338$
\item $p_4 = 0.36836$
\item $\tilde{\phi}(u) = p_2 \phi(u) \\
 = 0.14957 u - 0.0478 u^2 + 0.45 u^3 \textrm{ for } 0 \leq u \leq 0.23 \\
 = -0.00484 + 0.1995 u - 0.1067 u^2 + 0.158 u^3 \textrm{ for } 0.23 < u \leq 6/11\\
 = 0.47639 + 0.21685 u - 0.02388 u^2 - 0.021 u^3 \textrm{ for } 6/11 < u \leq 0.61\\
 = 0.47368 + 0.2816 u - 0.18365 u^2 + 0.079 u^3 \textrm{ for } 0.61 < u \leq 0.77\\
 = 0.32195 + 0.75 u - 0.6476 u^2 + 0.2239 u^3 \textrm{ for } 0.77 < u \leq 1.$
\end{itemize}

\

\noindent
Computational verification confirms that with these parameters, the cut density is upper-bounded in all cases by $1.296445$, for $0 \leq u_1 \leq u_2 \leq 1$ such that $u_1 + u_2 \leq 1$ and $u_1, u_2$ are multiples of $\delta = 1/2^{16}$. Finally, we use the following (standard) argument to bound the error arising from the discretization.

\begin{lemma}
\label{lem:disc-error}
For any function $f:[x_0,x_0+\delta] \times [y_0,y_0+\delta] \rightarrow \RR$ such that $\secdiff{f}{x}, \secdiff{f}{y} \geq -d$,
\begin{align*}
& \max_{x_0 \leq x \leq x_0+\delta, y_0 \leq y \leq y_0+\delta} f(x,y) \\
\leq &  \max \{f(x_0,y_0), f(x_0+\delta,y_0), f(x_0,y_0+\delta), f(x_0+\delta,y_0+\delta)\}
 + \frac14 d \delta^2.\end{align*}
\end{lemma}

\begin{proof}
Suppose the maximum of $f$ is attained at $(x^*,y^*)$. Consider the function $g(x,y) = f(x,y) + \frac12 d (x-x_0- \frac12 \delta)^2 + \frac12 d (y-y_0- \frac12 \delta)^2$.  By assumption, $\secdiff{g}{x} = \secdiff{f}{x} + d \geq 0$ and $\secdiff{g}{y} = \secdiff{f}{y} + d \geq 0$, i.e.~$g$ is convex along both axis-parallel directions. By convexity, we have
\begin{align*} g(x^*,y^*) \leq & \frac{x_0+\delta-x^*}{\delta} g(x_0,y^*) + \frac{x^*-x_0}{\delta} g(x_0+\delta,y^*) \\ \leq & \max \{ g(x_0, y^*), g(x_0+\delta,y^*) \}.\end{align*}
Repeating the same argument in the $y$-coordinate, we obtain
\begin{align*} 
 g(x^*,y^*) \leq & \max \{ g(x_0,y_0), g(x_0+\delta,y_0), 
g(x_0,y_0+\delta), g(x_0+\delta,y_0+\delta) \}.
\end{align*}
By the definition of $g$, we have $f(x,y) \leq g(x,y) \leq f(x,y) + \frac14 d \delta^2$ for $(x,y) \in [x_0,x_0+\delta] \times [y_0,y_0+\delta]$.
This implies that
\begin{align*} f(x^*,y^*) \leq & \max \{ f(x_0,y_0), f(x_0+\delta,y_0), f(x_0,y_0+\delta), f(x_0+\delta,y_0+\delta) \}  + \frac14 d \delta^2.\end{align*}
\end{proof}

We apply this lemma to the function $d_{12}(\bu)$ from Corollary~\ref{cor:1.3-total-density}. It can be verified that for all $\bu \in [0,1]^2$, $\secdiff{d_{12}}{u_1}, \secdiff{d_{12}}{u_2} \geq -16$. (The second derivatives are mostly easy to evaluate, except for the expressions for the Independent Thresholds Rounding Scheme --- appearing with a multiplier of $p_4$; we have verified these by {\em Mathematica}.) By Lemma~\ref{lem:disc-error}, the discretization error can be bounded by $\frac14 d \delta^2 = 4 \cdot 2^{-32} = 2^{-30} < 10^{-9}$. This proves the following theorem.

\begin{theorem}
\label{thm:1.2965}
Algorithm~\ref{alg:independent} with parameters as above provides a $1.2965$-approximation for the Multiway Cut problem.
\end{theorem}

\paragraph{Acknowledgment}
We would like to thank T. S. Jayram for bringing the references to exchangeble sequences (\cite{dF37,D77,DF80}) to our attention,
and David Aldous for a helpful discussion including pointing out the papers \cite{Kingman78} and \cite{TW98}.

\bibliography{multiway-cut-bib}
\bibliographystyle{plainnat}

\appendix
\section{Pairwise realizable distributions}
\label{sec:pairwise}

Here, we address a question that came up in our search for the best distribution of thresholds:

\

{\em Given a joint distribution $\rho$ of two random variables $(X,Y)$ over a domain $\cD$, is it possible to design an arbitrarily large number of random variables $X_1,X_2,\ldots,X_k$ such that for every pair $i \neq j$, the distribution of $(X_i, X_j)$ is the same distribution $\rho$?}

\begin{definition}
A probability distribution $\rho$ over $\cD \times \cD$ is {\em pairwise realizable}, if for every integer $k$, there exist $k$ random variables $X_1,X_2,\ldots,X_k$ such that for each pair $i \neq j$, the joint distribution of $(X_i,X_j)$ is $\rho$.
\end{definition}

In particular, we require that the distribution of $(X_i,X_j)$ is the same as the distribution of $(X_j,X_i)$, so $\rho$ should be symmetric. But even for symmetric distributions, this is not always possible. For example, we cannot design 3 random variables $X_1,X_2,X_3$ such that for each pair, we have $(X_i=0,X_j=1)$ or $(X_i=1,X_j=0)$ with probability $1/2$.
Here, we present a necessary and sufficient condition for a distribution to be pairwise realizable. To avoid technical complications, we restrict ourselves to discrete domains $\cD$.

\begin{theorem}
\label{thm:pairwise-realizable}
A probability distribution $\rho$ over a discrete domain $\cD \times \cD$ is pairwise realizable if and only if $\rho$ is a convex combination of symmetric product distributions:
$$ \rho(a,b) = \sum_{s=1}^{r} \alpha_s p_s(a) p_s(b), $$
where $p_s: \cD \rightarrow [0,1]$, $\sum_{a \in \cD} p_s(a) = 1$, $\alpha_s \in [0,1]$, and $\sum_{s=1}^{r} \alpha_s = 1$.
\end{theorem}

As pointed out to us by David Aldous, this statement (although apparently not stated explicitly) can be also derived from the theory of {\em exchangeable sequences} \citep{dF37,D77}. An infinite sequence of random variables $X_1, X_2, X_3,\ldots$ is exchangeable if for every $n$ and a permutation $\pi$ on $[n]$, $(X_1, X_2,$ $\ldots,X_n)$ and $(X_{\pi(1)}, X_{\pi(2)},$ $\ldots,$ $X_{\pi(n)})$ have the same distribution. De Finetti's Theorem states that the distribution of every exchangeable sequence is a convex combination of distributions of sequences where $X_1,X_2,\ldots$ are identical and independent \citep{dF37}.
Our statement can also be viewed as a result about infinite sequences of random variables: If $X_1,X_2,\ldots$ is an infinite sequence such that every {\em pair} has the same joint distribution, then this distribution is a convex combination of distributions of two identical and independent random variables. 
This can be derived from known results in the theory of exchangeability (using e.g., \citep{Kingman78,DF80,TW98}); here we give a direct proof using Farkas' lemma (or equivalently the convex separation theorem).

\

\begin{proof}[Theorem~\ref{thm:pairwise-realizable}]
The easy direction first:
If $\rho(a,b) = \sum_{s=1}^{r} \alpha_s p_s(a) p_s(b)$ as above, we can generate arbitrarily many random variables as follows: With probability $\alpha_s$, we generate $k$ independent random variables $X_1,X_2,\ldots,X_k$, each with probability distribution $p_s$. It is easy to see that for each $i \neq j$, the joint distribution of $(X_i,X_j)$ is $\rho$.

The reverse direction relies on the convex separation theorem. 
Define $\cS$ to be the convex hull of symmetric product distributions on $\cD \times \cD$, which is exactly the set of all distributions $\rho$ satisfying the conclusion of the theorem:
\begin{eqnarray*}
\cS & = & \Big\{ \rho: \cD \times \cD \rightarrow [0,1] \ \Big| \ \rho(a,b) = \sum_{s=1}^{r} \alpha_s p_s(a) p_s(b), 
  \sum_{a \in \cD} p_s(a) = 1, p_s(a) \geq 0, \sum_{s=1}^{r} \alpha_s = 1, \alpha_s \geq 0 \Big\}.
\end{eqnarray*}
We note that $\cS$ is a convex compact set in $\RR^{\cD \times \cD}$.
Hence, if $\tilde{\rho}: \cD \times \cD \rightarrow [0,1]$ is a distribution outside of $\cS$, then by the convex separation theorem there is $y: \cD \times \cD \rightarrow \RR$ and $\lambda_1 < \lambda_2$ such that
\begin{itemize}
\item $\sum_{a,b \in \cD} y(a,b) \tilde{\rho}(a,b) = \lambda_1$,
\item $\forall \rho \in \cS; \sum_{a,b \in \cD} y(a,b) \rho(a,b) \geq \lambda_2$.
\end{itemize}
In particular, for every distribution $p: \cD \rightarrow [0,1], \sum_{a \in \cD} p(a) = 1$, we have
\begin{equation}
\label{eq:lambda_2-bound}
 \sum_{a,b \in \cD} y(a,b) p(a) p(b) \geq \lambda_2.
\end{equation}
Assume for a contradiction that for any $k$, there are random variables $X_1,X_2,\ldots,X_k$ such that $ \forall i \neq j$ the distribution of $(X_i,X_j)$ is $\tilde{\rho}$. In particular, we choose $k > \frac{1}{\lambda_2-\lambda_1} (\max_{a \in \cD} y(a,a) - \lambda_1)$. (The right-hand side is at least $1$ by equation (\ref{eq:lambda_2-bound}).) We consider the following quantity:
\begin{equation}
\label{eq:magic}
 \sum_{a_1,a_2,\ldots,a_k \in \cD} \Pr[X_1=a_1,X_2=a_2,\ldots,X_k=a_k] \sum_{i \neq j} y(a_i,a_j).
\end{equation}
First, fix any choice of $a_1,a_2,\ldots,a_k \in \cD$ and consider the sum $\sum_{i \neq j} y(a_i,a_j)$.
Define $p(a) = \frac{1}{k} |\{ i \in [k]: a_i = a\}|$.
Then we have
\begin{eqnarray*}
\sum_{i \neq j} y(a_i,a_j) & = & \sum_{i,j=1}^{k} y(a_i,a_j) - \sum_{i=1}^{k} y(a_i,a_i) 
  = k^2 \sum_{a,b \in \cD} p(a) p(b) y(a,b) - k \sum_{a \in \cD} p(a) y(a,a).
\end{eqnarray*}
By the properties of $y(a,b)$, we get
\begin{eqnarray*}
 \sum_{i \neq j} y(a_i,a_j) & = & k^2 \sum_{a,b \in \cD} p(a) p(b) y(a,b) - k \sum_{a \in \cD} p(a) y(a,a) \\
 & \geq & k^2 \lambda_2 - k \sum_{a \in \cD} p(a) y(a,a) \\
 & \geq & k^2  \lambda_2 - k \, \max_{a \in \cD} y(a,a) .
\end{eqnarray*}
Now we use the fact that $k > \frac{1}{\lambda_2-\lambda_1} (\max_{a \in \cD} y(a,a) - \lambda_1)$, and therefore
$k \lambda_2 > k \lambda_1 + (\max_{a \in \cD} y(a,a) - \lambda_1) = (k-1) \lambda_1 + \max_{a \in \cD} y(a,a)$.
We obtain
\begin{eqnarray*}
 \sum_{i \neq j} y(a_i,a_j)  & > & k (k-1) \lambda_1.
\end{eqnarray*}
To summarize, expression (\ref{eq:magic}) which is a convex combination of such terms, is lower-bounded strictly by
\begin{equation}
\label{eq:magic-1}
\sum_{a_1,a_2,\ldots,a_k \in \cD} \Pr[X_1=a_1,X_2=a_2,\ldots,X_k=a_k] \sum_{i \neq j} y(a_i,a_j)  
 > k (k-1) \lambda_1.
 \end{equation}
On the other hand, by switching the sums, we have
\begin{eqnarray*}
& & \sum_{a_1,a_2,\ldots,a_k \in \cD} \Pr[X_1=a_1,X_2=a_2,\ldots,X_k=a_k] \sum_{i \neq j} y(a_i,a_j) \\
&= & \sum_{i \neq j} \sum_{a_i,a_j \in \cD} y(a_i,a_j) \hspace{-0.2in} \sum_{(a_\ell \in \cD: \ell \neq i,j)}
   \Pr[X_1=a_1,X_2=a_2,\ldots,X_k=a_k] \\
&= & \sum_{i \neq j} \sum_{a_i,a_j \in \cD} y(a_i,a_j) \Pr[X_i=a_i, X_j=a_j] \\
&= & \sum_{i \neq j} \sum_{a_i,a_j \in \cD} y(a_i,a_j) \tilde{\rho}(a_i,a_j) =  k(k-1) \lambda_1
\end{eqnarray*}
again by the properties of $y(a,b)$. This is a contradiction with (\ref{eq:magic-1}).
\end{proof}

\section{A tight example for $\frac{3+\sqrt{5}}{4}$}
\label{sec:tight-example}

Here we show the following lower bound.

\begin{theorem}
No combination of the Single Threshold Rounding Scheme (with any distribution) and the Exponential Clocks Rounding Scheme can provide an approximation factor better than $(3+\sqrt{5})/4$. 
\end{theorem}

We show this by considering a game between the algorithm and  an adversary. The strategy space of the adversary is to pick an edge in the simplex. The strategy space of the algorithm is to pick a partition of the simplex into $k$ parts of the following type:
\begin{enumerate}
\item Exponential Clocks: a random partition generated by Algorithm~\ref{alg:exp-clock}, or
\item Single Threshold: a partition generated by Algorithm~\ref{alg:single-threshold} for any fixed value of $\theta$
\end{enumerate}
and any probability distribution over the strategies above.
The game is a zero-sum game where if the endpoints of the edge picked by the adversary belong to two different parts of the partition picked by the algorithm, then the algorithm pays to the adversary a cost of $1$ divided by the length of the edge (where recall that the length of the edge is defined to be $\frac12 \times$ the $L_{1}$-distance between its endpoints). If the edge belongs to one part of the partition, there is $0$ payoff to both players.

Clearly, if there is an algorithm in the strategy space that achieves a cut density bounded by $\alpha$, then this implies a strategy for the algorithm player that pays at most $\alpha$ in expectation against any adversary. We present a strategy for the adversary such that the algorithm has to pay at least $\frac{3+\sqrt{5}}{4} - O(\frac{1}{k})$ in expectation. I.e., a cut density better than $\frac{3+\sqrt{5}}{4} - O(\frac{1}{k})$  cannot be achieved by this type of algorithm.

\subsection{Probability distribution over edges}

Here we define a strategy for the adversary. 
For each \emph{ordered} pair of terminals $(i,j)$ ($i, j \in [k]$, $i\ne j$), we have three sets of edges: $A_{ij}$, $B_{ij}$ and $C_{ij}$. The edges in each of the sets $A_{ij}$, $B_{ij}$ and $C_{ij}$ have the property that their endpoints differ only in the coordinates $i$ and $j$. Hence, the edges ${\bf (v,v')}$ in these sets are of the form ${\bf v}=(u_{1}, \cdots, u_{k})$ and ${\bf v'}=(u_{1}, \cdots, u_{i}-\epsilon, \cdots, u_{j}+\epsilon, \cdots u_{k})$. In all the three sets, we shall have $\forall k\ne i,j, ~ u_{k}=(1-u_{i}-u_{j})/(k-2)$. Hence, each edge is defined by the values of $u_{i}$ and $u_{j}$.
\begin{itemize}[noitemsep,topsep=0pt,parsep=0pt,partopsep=0pt]
\item $A_{ij}$ consists of all edges such that $u_{i}=x$ for some $x\in[3b, 1]$ and $u_{j}=0$.
\item $B_{ij}$ consists of all edges such that $u_{i}+2u_{j}=3b$ with $b \le u_{i} \le 3b$ (and hence, $0\le u_{j}\le b$).
\item $C_{ij}$ consists of a single edge defined by $u_{i}=1$ and $u_{j}=0$.
\end{itemize}
This completes the description of the sets $A_{ij}$, $B_{ij}$ and $C_{ij}$ except for setting the values of the parameters $b$ and $\epsilon$. The three sets are depicted in Figure~\ref{fig:tight-example}. We shall set $b= \sqrt{5}-2$ (as in Section~\ref{sec:1.309}) and $\epsilon = (1-2b)/(k-2)$.

\begin{figure}[h]

\centering

\pgfmathsetmacro{\Scale}{3}

\newcommand{\flauntFont}[1]{{\large \bf #1}}

\pgfmathsetmacro{\eps}{0.05}

\pgfmathsetmacro{\PointB}{sqrt(5)-2}
\pgfmathsetmacro{\BFactor}{(1-3*\PointB)*0.5}

\pgfmathsetmacro{\terminalOneThreeBX}{\BFactor}
\pgfmathsetmacro{\terminalOneThreeBY}{\BFactor*2}

\pgfmathsetmacro{\terminalTwoThreeBX}{1 - \terminalOneThreeBX}
\pgfmathsetmacro{\terminalTwoThreeBY}{\terminalOneThreeBY}

\pgfmathsetmacro{\BBX}{0.5}
\pgfmathsetmacro{\BBY}{0.5*(1+\PointB)}

\def\DrawEdges#1#2#3#4#5{
\pgfmathsetmacro{\stepsizeinverse}{100}

\pgfmathsetmacro{\increment}{\eps*#5}

\pgfmathsetmacro{\startpoint}{#1*\stepsizeinverse}
\pgfmathsetmacro{\endpoint}{#3*\stepsizeinverse}
\pgfmathsetmacro{\Slope}{(#4-#2)/(#3-#1)}
\pgfmathsetmacro{\Intersection}{#2 - \Slope* #1}

\foreach \x [evaluate=\x as \pointx using ((\x)/\stepsizeinverse, evaluate=\x as \pointy using ((\Slope*\pointx)+\Intersection) ]in {\startpoint,...,\endpoint}	
{
\draw[-] (\pointx , \pointy) -- (\pointx+\increment , \pointy);
}
}

\begin{tikzpicture}[scale=\Scale]

\draw[dashed,-](0,0) -- (1,0);
\draw[dashed,-](0,0) -- (1/2,1);
\draw[dashed,-](1/2,1) -- (1,0);
\draw[dashed,-](0.25-0.25*\PointB,0) -- (5/8-1/8*\PointB,3/4+1/4*\PointB) node[right]{$(u_1=0, u_2=b)$};
\draw[dashed,-](0.75+0.25*\PointB,0) -- (3/8+1/8*\PointB,3/4+1/4*\PointB) node[left]{$(u_1=b, u_2=0)$};

\draw[] (0.0,0.0) node[left] {$(u_1=1, u_2=0)$};
\draw[] (1.0,0.0) node[right] {$(u_1=0, u_2=1)$};
\draw[] (0.15,0.3) node[left] {$(u_1=3b, u_2=0)$};
\draw[] (0.85,0.3) node[right] {$(u_1=0, u_2=3b)$};
\draw[] (0.5,1.0) node[above] {$(u_1=0, u_2=0)$};
\draw[] (0.1,0.15) node[right] {\flauntFont{$A_{12}$}};
\draw[] (0.9,0.15) node[left] {\flauntFont{$A_{21}$}};
\draw[] (0.25,0.55) node[right] {\flauntFont{$B_{12}$}};
\draw[] (0.75,0.55) node[left] {\flauntFont{$B_{21}$}};
\draw[] (0.05,0) node[below] {\flauntFont{$C_{12}$}};
\draw[] (0.95,0) node[below] {\flauntFont{$C_{21}$}};

\draw[line width=3pt] (0,0) -- (\eps,0);
\draw[line width=3pt] (1,0) -- (1-\eps,0);

\DrawEdges{0}{0}{\terminalOneThreeBX}{\terminalOneThreeBY}{1}
\DrawEdges{\terminalOneThreeBX}{\terminalOneThreeBY}{\BBX}{\BBY}{1}
\DrawEdges{\terminalTwoThreeBX}{\terminalTwoThreeBY}{\BBX}{\BBY}{-1}
\DrawEdges{1}{0}{\terminalTwoThreeBX}{\terminalTwoThreeBY}{-1}

\end{tikzpicture}

\vspace{-0.2in}
\caption{Tight example: edges of type $(1,2)$.}
\label{fig:tight-example}
\end{figure}
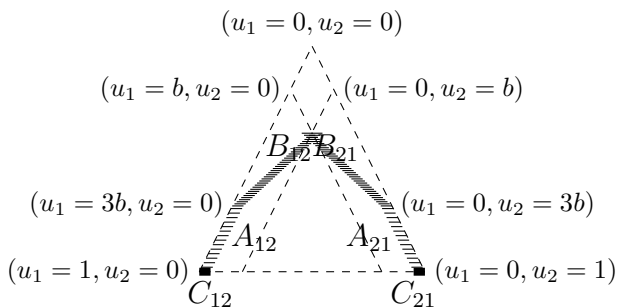

We now define the probability distribution over $\bigcup_{(i,j)} (A_{ij} \cup B_{ij} \cup C_{ij})$ that defines the strategy of the adversary. Starting with $C_{ij}$, the single edge is chosen with probability $\frac{2(1-2b)}{(1-b)(k-2)}~\frac{1}{k(k-1)}$. Hence, the total probability mass over the edges in $\bigcup_{(i,j)}C_{ij}$ is $\frac{2(1-2b)}{(1-b)(k-2)}$.

The adversary has a uniform distribution over edges in $\bigcup_{(i,j)}A_{ij}$. Similarly, the adversary has a uniform distribution over edges in $\bigcup_{(i,j)}B_{ij}$. The total probability mass of the edges in $\bigcup_{(i,j)}A_{ij}$ is $\frac{1-3b}{1-b}~(1-\frac{1}{k-2})$, and the total probability mass of the edges in $\bigcup_{(i,j)}B_{ij}$ is $(\frac{2b}{1-b} - \frac{1}{k-2})$. Note that there are $k(k-1)$ ordered pairs $(i,j)$, hence for a particular ordered pair $(i,j)$, the sets $A_{ij}$ and $B_{ij}$ carry a total probability mass of $\frac{(1-3b)}{(1-b)}(1-\frac{1}{k-2}) \frac{1}{k(k-1)} $ and $(\frac{2b}{1-b} - \frac{1}{k-2}) \frac{1}{k(k-1)}$ respectively.

We defer the analysis of the max-min value of the game to the full version of the paper.

It can be verified that the probabilities add up to $$\frac{2(1-2b)}{(1-b)(k-2)} + \frac{1-3b}{1-b}\left(1-\frac{1}{k-2}\right) + \left(\frac{2b}{1-b} - \frac{1}{k-2}\right) = 1.$$

In the following two subsections, we show that any strategy adopted by the algorithm will incur a cost of at least $\frac{3+\sqrt{5}}{4} - O(\frac{1}{k})$.
Note that length of every edge used in the probability distribution by the adversary is $\epsilon$. Therefore, we need to show that the probability of the edge being cut by any strategy of the algorithm is at least $\left( \frac{3+\sqrt{5}}{4} - O(\frac{1}{k}) \right) \epsilon$. 
We remark that by our choice of $b = \sqrt{5}-2$, we have $\frac{3+\sqrt{5}}{4} = \frac{1}{1-b}$.

\subsection{Performance of the Exponential Clocks \\ Rounding Scheme}

For any edge $({\bf v,v'})$ which differ only on coordinates $i$ and $j$, and is of length $\epsilon$, we know from \cite{BNS13} that the probability that a random partition from the Exponential Clocks Rounding Scheme cuts the edge is $\frac{2-u_{i}-u_{j}+\epsilon}{1+\epsilon}~\epsilon \geq (2-u_i-u_j-\epsilon)~\epsilon$. Hence, the probability of an edge being cut in the set $A_{ij}$ (for an ordered pair $(i,j)$) is at least
\begin{align*}\frac{1-3b}{1-b} \left(1-\frac{1}{k-2}\right) \frac{1}{k(k-1)} \cdot \frac{1}{1-3b}~\int_{3b}^{1}(2-x-\epsilon) \mathrm{d}x \cdot \epsilon \\= \left( \frac{3}{2}(1-3b) - O\left(\frac{1}{k} \right) \right)~\frac{\epsilon}{k(k-1)}.\end{align*}
Similarly, the probability of an edge being cut in the set $B_{ij}$ is at least
\begin{align*}\left(\frac{2b}{1-b} - \frac{1}{k-2}\right)\frac{1}{k(k-1)} \cdot \frac{1}{b}~\int_{0}^{b}(2-(3b-x)-\epsilon) \mathrm{d}x \cdot \epsilon \\= \left(\frac{b}{1-b}~\left(4-5b\right) - O\left(\frac{1}{k}\right) \right)~\frac{\epsilon}{k(k-1)}~.\end{align*}
The probability of an edge being cut in $C_{ij}$ is $\frac{2(1-2b)}{(1-b)(k-2)}\frac{\epsilon}{k(k-1)}$, which does not contribute significantly to the total probability; since we are interested in a lower bound, we can ignore the contribution of $C_{ij}$ here.

Adding up over all $i \neq j$, we obtain that the total probability of an edge being cut by the algorithm is at least
 $$\left(\frac{3}{2}(1-3b) + \frac{b}{1-b}~\left(4-5b\right) - O\left(\frac{1}{k}\right) \right)~\epsilon$$
which is equal to $\left(\frac{3+\sqrt{5}}{4} - O(\frac{1}{k}) \right)~\epsilon$ after plugging in $b=\sqrt{5}-2$. 

\subsection{Performance of a partition induced by a single threshold}

We now consider a partition induced by a choosing a single threshold $\theta$ and a random permutation $\sigma$ of the terminals. Since the distribution of the edges over the terminals is symmetric, we can consider the case when $\sigma$ is the identity permutation. 

Define an edge to be \emph{captured} by a terminal if both endpoints of the edge are assigned to this  terminal. Similarly, define an edge to be \emph{cut} by a terminal if one of the endpoints is assigned to the terminal but the other endpoint is not.

Before we delve into the analysis, we would like to make an observation. For any ordered pair $(i,j)$, in both $A_{ij}$ and $B_{ij}$, we know that $\forall k \ne i,j$, $u_{k}=(1-u_{i}-u_{j})/(k-2)$. Since $u_i+u_j \geq 2b$ for all edges in $A_{ij} \cup B_{ij}$, this means that the values of $u_{k}$ ($k\ne i,j$) among all edges in $A_{ij}$ and $B_{ij}$ are upper-bounded by $(1-2b)/(k-2) = \epsilon$. For the edges in $C_{ij}$, we have $u_i + u_j = 1$, and hence all the remaining coordinates are $0$.
This implies the following two observations.

\begin{observation}
\label{obs:ABuntouchedbyk}
For $\theta \in (\epsilon,1]$, for any ordered pair $(i,j)$, the probability that an edge in $A_{ij}$ or $B_{ij}$ is either captured or cut by a terminal $k\ne i,j$ is zero. This holds irrespective of the permutation chosen over the terminals.
\end{observation}

\begin{observation}
\label{obs:Cuntouchedbyk}
For $\theta \in (0,1]$, for any ordered pair $(i,j)$, the probability that an edge in $C_{ij}$ is either captured or cut by a terminal $k\ne i,j$ is zero. This holds irrespective of the permutation chosen over the terminals.
\end{observation}

We break the analysis into several cases depending on the value of $\theta$.
We shall use the following in our analysis below.
\begin{itemize}
\item Observation~\ref{obs:ABuntouchedbyk} applies to the range of $\theta$ dealt with in Sections~\ref{sec:3bthru1}, \ref{sec:bthru3b}, and \ref{sec:lowthrub}, and hence while calculating the probability of an edge from $A_{ij}$ and $B_{ij}$ being cut, we can focus only on whether terminals $i$ or $j$ cut the edge.
\item Observation~\ref{obs:Cuntouchedbyk} applies to the entire range of $\theta$, and in particular we shall use it in Sections~\ref{sec:lowthrub} and \ref{sec:verytop} by focusing only on whether terminals $i$ and $j$ cut the edge when dealing with the edge from $C_{ij}$.
\item In the algorithm, when using the strategy of a single threshold, we note that only the first $k-1$ terminals in the random permutation get assigned vertices according to the threshold. The last terminal gets assigned all the remaining unassigned vertices. Hence, in the following analysis, at several points, we shall sum over only $k-1$ vertices (instead of $k$).
\end{itemize}
Finally, we remind the reader we shall be assuming that $\sigma$ is the identity permutation in the analysis below; since the distribution of the edges is symmetric over the terminals, the choice of the permutation does not make a difference in the analysis. W recall that $b = \sqrt{5}-2$ and  $\frac{1}{1-b} = \frac{3+\sqrt{5}}{4}$ which will be useful below.

\subsubsection{Case $1-\epsilon < \theta \leq 1$}
\label{sec:verytop}
In this case, for each $i,j\in [k], ~i\ne j$, the edge in $C_{ij}$ will be cut by terminal $i$. Since each $C_{ij}$ is chosen with probability $\frac{2(1-2b)}{(1-b)(k-2)} \frac{1}{k(k-1)}$, the total probability of an edge being cut is at least $\frac{2(1-2b)}{(1-b)(k-2)}$. Note that by the choice of $\epsilon$, this quantity is equal to $\frac{2}{1-b}~\epsilon = \frac{3+\sqrt{5}}{2} \epsilon$.

\subsubsection{Case $3b < \theta \le 1-\epsilon$}
\label{sec:3bthru1}
It is easy to see in this case that for each $i \in [k-1]$, the terminal $i$ will cut the edges in $A_{ij}$ (for all $j\in [k] \setminus \{i\}$) where $u_{i}-\epsilon < \theta \le u_{i}$. This is an $\epsilon$-size interval among the edges in $A_{ij}$ parameterized by $u_i$. Since the probability density of choosing an edge with given $u_i$ in $A_{ij}$ is $\frac{1}{1-b} (1-\frac{1}{k-2}) \frac{1}{k(k-1)}$, the probability that terminal $i$ cuts an edge in $A_{ij}$ is $\frac{1}{1-b}(1-\frac{1}{k-2}) \frac{1}{k(k-1)}\cdot \epsilon$. Summing over all $A_{ij}$'s ($i \in [k-1], j\in [k] \setminus \{i\}$), we get that the total probability of an edge being cut is $\frac{1}{1-b}(1-\frac{1}{k-2})(1-\frac{1}{k})~\epsilon = \frac{3+\sqrt{5}}{4}(1-O(\frac{1}{k}))~\epsilon$.

\subsubsection{Case $b < \theta < 3b-\epsilon$}
\label{sec:bthru3b}
In this case, each $i \in [k-1]$ cuts the edges in $B_{ij}$ (for all $j\in [k] \setminus \{i\}$) where $u_{i}-\epsilon < \theta \le u_{i}$. Again, this is an $\epsilon$-size interval among the edges in $B_{ij}$, in terms of the parameter $u_i$. Since the probability density of choosing an edge with given $u_i$ in $B_{ij}$ is $(\frac{1}{1-b} - \frac{1}{2b(k-2)}) \frac{1}{k(k-1)}$, the probability of an edge being cut is $(\frac{1}{1-b} - \frac{1}{2(k-2)b})~\frac{1}{k(k-1)}~\epsilon$. Summing over all $B_{ij}$'s ($i \in [k-1], j\in [k] \setminus \{i\}$), we get that the total probability of an edge being cut is $(\frac{1}{1-b} - \frac{1}{2(k-2)b})~(1-\frac{1}{k})~\epsilon
 = (\frac{3+\sqrt{5}}{4} - O(\frac{1}{k}))~\epsilon$.

\subsubsection{Case $3b-\epsilon \le \theta \le 3b$}
This case is essentially a transition between Case~\ref{sec:3bthru1} and \ref{sec:bthru3b}. Here, both edges in $A_{ij}$ and $B_{ij}$ can be cut by terminal $i$, with probabilities that depend on the value of $\theta$. As $\theta$ moves from $3b-\epsilon$ to $3b$, the interval of edges cut in $A_{ij}$ increases and the interval of edges cut in $B_{ij}$ decreases linearly.
It can be verified that the probability of being cut is a convex combination of the probabilities in Cases~\ref{sec:3bthru1} and \ref{sec:bthru3b}, with a linear transition from Case~\ref{sec:3bthru1} at $\theta=3b$ to Case~\ref{sec:bthru3b} at $\theta=3b-\epsilon$, hence always bounded by $(\frac{3+\sqrt{5}}{4} - O(\frac{1}{k}))~\epsilon$.

\subsubsection{Case $\epsilon < \theta \le b$}
\label{sec:lowthrub}
In this case, each terminal $i \in [k-1]$ cuts the edges in $B_{ji}$ (for all $j>i$) where $u_{i} \le \theta \le u_{i}+\epsilon$. Please note the two ways this case differs from Section~\ref{sec:bthru3b}. First the edges in $B_{ji}$ are being cut by terminal $i$ (and not by terminal $j$) and only when $j>i$ (i.e., terminal $j$ is considered after terminal $i$, otherwise the edge would have been captured by terminal $j$). The probability density of choosing edges in $B_{ji}$ in terms of the parameter $u_i$ is $(\frac{2}{1-b} - \frac{1}{(k-2)b}) \frac{1}{k(k-1)}$, twice as large as the probability density in terms of parameter $u_j$. Therefore, for any $B_{ji}$, with $i\in [k-1]$ and $j>i$, the probability of an edge being cut is $(\frac{2}{1-b} - \frac{1}{(k-2)b})~\frac{1}{k(k-1)}~\epsilon$. Now we are summing only over $i \in [k-1]$ and $j>i$, which leads again to a total probability of an edge being cut $(\frac{1}{1-b} - \frac{1}{2(k-2)b})~\epsilon  = (\frac{3+\sqrt{5}}{4} - O(\frac{1}{k}))~\epsilon$.

\subsubsection{Case $0 < \theta \le \epsilon$}
\label{sec:verybottom}
This is the only case in which we have to consider the possibility of an edge being captured by a terminal other than $i,j$. (See Observations~\ref{obs:ABuntouchedbyk}~and~\ref{obs:Cuntouchedbyk} above.) Here we count only the contribution of edges in $C_{ij}$ (which can never be captured by $k \neq i,j$ by Observation~\ref{obs:Cuntouchedbyk}). 
In this case, for each $i\in [k]$ and $j>i$, the edge in $C_{ji}$ is cut by terminal $i$. Hence, the total probability of an edge being cut in $\bigcup_{(i,j):j>i}C_{ij}$ is $\frac{1-2b}{(1-b)(k-2)}$. Note that by the choice of $\epsilon$, this quantity is equal to $\frac{1}{1-b}~\epsilon = \frac{3+\sqrt{5}}{4}~\epsilon$.

\begin{table}[h]
\begin{center}
\begin{tabular}{|c|c|c|c|}
\hline
\multirow{2}{*}{Range of $\theta$}&\multicolumn{2}{|c|}{Cost incurred by}&\multirow{2}{*}{Cost}\\\cline{2-3}
&Edges cut in&by Terminal&\\\hline
$[1-\epsilon, 1]$ & $C_{ij}$ & $i$ & $\frac{3+\sqrt{5}}{2}$\\\hline
$[3b, 1-\epsilon]$ & $A_{ij}$ & $i$ & $\frac{3+\sqrt{5}}{4} - O(\frac{1}{k})$\\\hline
$[b, 3b]$ & $B_{ij}$ & $i$ & $\frac{3+\sqrt{5}}{4} - O(\frac{1}{k})$\\\hline
$[\epsilon, b]$ & $B_{ji}$ & $i$ ($j \succ i$)$^{\dagger}$ & $\frac{3+\sqrt{5}}{4}  - O(\frac{1}{k})$\\\hline
$(0,\epsilon]$ & $C_{ji}$ & $i$ & $\frac{3+\sqrt{5}}{4}$\\\hline
\end{tabular}
\end{center}
\caption{This table summarizes the expected cost paid by the algorithm depending on the value of $\theta$ in the single-threshold strategy. \newline $^{\dagger}$This indicates that terminal $j$ must occur after terminal $i$ in the random permutation over terminals chosen by the algorithm.}
\end{table}
\section{A tight example for $\frac{10+4\sqrt{3}}{13}$}
\label{sec:1.302-tight}
Here we prove the following lower bound, matching our $\frac{10+4\sqrt{3}}{13}$-approximation from Section~\ref{sec:1.302}.
\vspace{-0.1in}
\begin{theorem}
For any combination of the Exponential Clocks Rounding Scheme and schemes with $k$ threshold cuts, one for each terminal (and an arbitrary joint distribution), as long as the analysis of threshold cuts considers only the thresholds $\theta_i, \theta_j$ for edges of type $(i,j)$ (and not the possibility of being captured by another terminal), it cannot achieve a factor better than $\frac{10+4\sqrt{3}}{13}$.
\end{theorem}

More precisely, what this means that the analysis takes into account the possibility of cutting an edge $(i,j)$ by thresholds $\theta_i, \theta_j$ or allocating the edge fully to terminal $i$ or $j$, but not the possibility of allocating the edge fully to another terminal. This is the way we analyze our algorithm in Section~\ref{sec:1.302} where this factor is achieved. 

As in Section~\ref{sec:tight-example}, we define a probability distribution of edges in the simplex such that any partition strategy of the following type has to pay at least $\frac{10+4\sqrt{3}}{13}$ in expectation:
\begin{enumerate}[noitemsep,topsep=0pt,parsep=0pt,partopsep=0pt]
\item Exponential Clocks: a random partition generated by Algorithm~\ref{alg:exp-clock}, or
\item Simple Thresholds: any sequence of $k$ threshold cuts, one for each terminal (such as the Single Threshold or Descending Thresholds Rounding Scheme).
\end{enumerate}

Here, we modify the game to reflect the fact that we are considering analysis that depends only on the two coordinates $u_i,u_j$ for edges of type $(i,j)$: if the partition player uses thresholds $\theta_i,\theta_j$ in these coordinates, and say $\theta_i$ is applied first, he pays whenever the edge is cut in coordinate $u_i$ or the edge is cut in coordinate $u_j$ and {\em not captured} in coordinate $u_i$ (i.e. other coordinates are not considered for the purposes of determining the payment). 

\subsection{Probability distribution over edges}

Let us define the distribution over edges of type $(1,2)$ (Figure~\ref{fig:tight-example-1-302}). Edges of type $(i,j)$ are distributed analogously.
We have $\alpha = \frac{-3 + 4\sqrt{3}}{13}$, $\gamma = \frac{19 - 8\sqrt{3}}{26}$ and $b = 2 \sqrt{3} - 3$.
The location of each edge satisfies $u_{3} = u_{4} = \cdots = u_{k} = (1 - u_{1} - u_{2})/(k-2)$, i.e. it is determined by $u_1$ and $u_2$.

\begin{itemize}[noitemsep,topsep=0pt,parsep=0pt,partopsep=0pt,leftmargin=*]
\item Region $R_{A}$: Uniform density in the region defined by $u_{1} \in [0,b]$ and $u_{2} \in [0,b]$ with a total probability mass of $\alpha$.
\item Region $R_{B1}$: Uniform density in the region defined by $u_{1} - \frac{1-2b}{b} u_{2} \ge b$ on the simplex with a total probability mass of $\alpha/4$.
\item Region $R_{B2}$: Uniform density in the region defined by $u_{2} - \frac{1-2b}{b} u_{1} \ge b$ on the simplex with a total probability mass of $\alpha/4$.
\item Region $R_{C1}$: Density proportional to $(u_{1} - (1-b))$ in the region defined by $u_{1} \in [1-b,1]$ and $u_{1} + u_{2} = 1$, with a total probability mass of $\frac{b}{1-b} (\alpha/4)$.
\item Region $R_{C2}$: Density proportional to $(u_{2} - (1-b))$ in the region defined by $u_{2} \in [1-b,1]$ and $u_{1} + u_{2} = 1$, with a total probability mass of $\frac{b}{1-b} (\alpha/4)$.
\item Region $R_{D1}$: Density proportional to $(1-b-u_{1})$ in the region defined by $u_{1} - \frac{1-2b}{b} u_{2} = b$ and $u_{1} \in [b,1-b]$,  with a total probability mass of $\frac{1-2b}{1-b} (\alpha/4)$.
\item Region $R_{D2}$: Density proportional to $(1-b-u_{2})$ in the region defined by $u_{2} - \frac{1-2b}{b} u_{1} = b$ and $u_{1} \in [b,1-b]$, with a total probability mass of $\frac{1-2b}{1-b} (\alpha/4)$.
\item Region $R_{E1}$: Uniform density in the region defined by $u_{1} \in [b,1]$ and $u_{2}=0$, with a total probability mass of $\gamma$.
\item Region $R_{E2}$: Uniform density in the region defined by $u_{2} \in [b,1]$ and $u_{1}=0$, with a total probability mass of $\gamma$.
\end{itemize}

We note that the total probability mass we have used is $\alpha + \frac{\alpha}{2} + \frac{b}{1-b} \frac{\alpha}{2} + \frac{1-2b}{1-b} \frac{\alpha}{2} + 2 \gamma = 2 \alpha + 2 \gamma = 1$.

\begin{figure}[h]

\centering

\pgfmathsetmacro{\Scale}{3}

\newcommand{\flauntFont}[1]{{ \bf #1}}

\pgfmathsetmacro{\eps}{0.05}

\pgfmathsetmacro{\B}{2*sqrt(3) - 3}

\pgfmathsetmacro{\sineSixty}{sin(60)}
\pgfmathsetmacro{\cosineSixty}{cos(60)}

\pgfmathsetmacro{\uOnex}{0}
\pgfmathsetmacro{\uOney}{0}
\pgfmathsetmacro{\uTwox}{1}
\pgfmathsetmacro{\uTwoy}{0}
\pgfmathsetmacro{\uThreex}{\cosineSixty}
\pgfmathsetmacro{\uThreey}{\sineSixty}
\pgfmathsetmacro{\bx}{\B * \cosineSixty}
\pgfmathsetmacro{\by}{\B * \sineSixty}
\pgfmathsetmacro{\OneMinusBx}{((1-\B) * \cosineSixty)}
\pgfmathsetmacro{\OneMinusBy}{((1- \B) * \sineSixty)}

\pgfmathdeclarefunction{customIfIfElse}{5}{%
\pgfmathparse{#1*#2 + !#1*(#3*#4 + !#3*#5)}%
}
\pgfmathdeclarefunction{customIfElse}{3}{\pgfmathparse{#1 * #2 + !#1 * #3}}

\def\DrawEdges#1#2#3#4#5#6#7#8{
{
\pgfmathsetmacro{\stepsizeinverse}{100}

\pgfmathsetmacro{\increment}{#5*\eps}

\pgfmathsetmacro{\startpoint}{#1*\stepsizeinverse}
\pgfmathsetmacro{\endpoint}{#3*\stepsizeinverse}

\pgfmathsetmacro{\Slope}{((#4-#2)/(#3-#1))}
\pgfmathsetmacro{\Intersection}{(#2 - \Slope* #1)}



\ifthenelse{\equal{#6}{1}}
{\edef\mya{#7}}{
\ifthenelse{\equal{#6}{-1}}
{\edef\mya{#7}}{\edef\mya{1}}
}

\foreach \x [evaluate=\x as \pointx using ((\x)/\stepsizeinverse), evaluate=\x as \pointy using ((\Slope * \pointx)+\Intersection)] in {\startpoint,...,\endpoint}	
{



\ifthenelse{\equal{#6}{1}}{
\pgfmathparse{\mya + #8};
\xdef\mya{\pgfmathresult};
}{;}

\ifthenelse{\equal{#6}{-1}}{
\pgfmathparse{\mya - #8};
\xdef\mya{\pgfmathresult};
}{;}

\draw[-, line width=\mya*1pt] (\pointx , \pointy) -- (\pointx+\increment , \pointy);
}
}
}

\begin{tikzpicture}[scale=\Scale]

\coordinate(uOne) at (\uOnex, \uOney);
\coordinate(uTwo) at (\uTwox, \uTwoy);
\coordinate(uThree) at (\uThreex, \uThreey);
\coordinate(A) at (\B, 0);
\coordinate(B) at (1 -\B, 0);
\coordinate(C) at (\OneMinusBx, \OneMinusBy);
\coordinate(D) at (1 - \OneMinusBx, \OneMinusBy);
\coordinate(E) at (0.5, 1 - 2*\B);

\draw[fill=gray!40] (uThree) -- (C) -- (E) -- (D) -- cycle;
\draw[fill=gray!50] (C) -- (A) -- (uOne) -- cycle;
\draw[fill=gray!50] (D) -- (B) -- (uTwo) -- cycle;

\draw[dashed,-] (uOne) -- (uTwo);
\draw[dashed,-] (uTwo) -- (uThree);
\draw[dashed,-] (uOne) -- (uThree);

\draw[dashed,-] (A) -- (D);
\draw[dashed,-] (B) -- (C);

\draw[] (uOne) node[left] {$(u_1=1, u_2=0)$};
\draw[] (uTwo) node[right] {$(u_1=0, u_2=1)$};
\draw[] (uThree) node[above] {$(u_1=0, u_2=0)$};
\draw[] (C) node[left] {$(u_{1}=b, u_{2}=0)$};
\draw[] (D) node[right] {$(u_{1}=0, u_{2}=b)$};

\draw[] (0.5,0.5) node[above]{\flauntFont{$R_{A}$}};
\draw[] (0.25,0.1) node[above]{\flauntFont{$R_{B1}$}};
\draw[] (0.75,0.1) node[above]{\flauntFont{$R_{B2}$}};
\draw[] (0.25,0) node[below]{\flauntFont{$R_{C1}$}};
\draw[] (0.75,0) node[below]{\flauntFont{$R_{C2}$}};
\draw[] (\B - 0.2, \OneMinusBy/2+ 0.1) node[right]{\flauntFont{$R_{D1}$}};
\draw[] (1 - \B - 0.05, \OneMinusBy/2+ 0.1) node[right]{\flauntFont{$R_{D2}$}};
\draw[] (\OneMinusBx/2, \OneMinusBy/2) node[left]{\flauntFont{$R_{E1}$}};
\draw[] (1 - \OneMinusBx/2, \OneMinusBy/2) node[right]{\flauntFont{$R_{E2}$}};

\DrawEdges{0}{0}{\OneMinusBx}{\OneMinusBy}{0.5}{0}{0}{0}
\DrawEdges{1}{0}{(1-\OneMinusBx)}{\OneMinusBy}{-0.5}{0}{0}{0}
\DrawEdges{\OneMinusBx}{\OneMinusBy}{\B}{0}{-0.5}{-1}{2}{0.1}
\DrawEdges{(1-\B)}{0}{(1-\OneMinusBx)}{\OneMinusBy}{0.5}{1}{0}{0.1}
\DrawEdges{0}{0}{\B}{0}{0.5}{-1}{5}{0.1}
\DrawEdges{(1-\B)}{0}{1}{0}{0.5}{1}{0.01}{0.1}

\end{tikzpicture}
\vspace{-0.2in}
\caption{Tight example: edges of type $(1,2)$ have been shown.}
\label{fig:tight-example-1-302}
\end{figure}
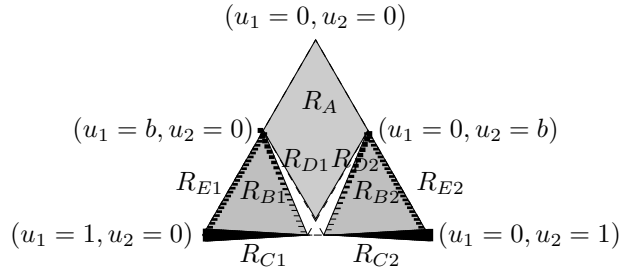

We defer the analysis to the full version of the paper.

\subsection{Performance of the Exponential Clocks \\Rounding Scheme}

For an edge of location $\bu$, the cut density is $2-u_1-u_2$ by Lemma~\ref{lem:density}. We compute the contribution of each region to the total expected cost of the partition. Since the expression if linear in $(u_1,u_2)$, the average over each region is equal to $2-c_1-c_2$ where $(c_1,c_2)$ is the center of mass of that region.

\subsubsection{Cut density in region $R_{A}$}
The center of mass of this diamond-shaped region is $(\frac{b}{2}, \frac{b}{2})$, and its probability mass is $\alpha$,  hence the expected cost is $\alpha (2 - \frac{b}{2} - \frac{b}{2}) = \alpha (2-b)$.

\subsubsection{Cut density in regions $R_{B1}$ and $R_{B2}$}
$R_{B1}$ is a triangle, with corners $(1,0), (b,0), (1-b,b)$. Hence, its center of mass is $(\frac{2}{3}, \frac{b}{3})$ and the expected cost is $\frac{\alpha}{4} (2 - \frac{1}{3} (2+b)) = \frac{\alpha}{12} (4 - b)$. By symmetry, the cost for region $R_{B2}$ is also $\frac{\alpha}{12} (4 - b)$.

\subsubsection{Cut density in regions $R_{C1}$ and $R_{C2}$}
Each edge in these regions has a cut density of 1. Hence, the cost for each of $R_{C1}, R_{C2}$ is $\frac{b}{1-b} \frac{\alpha}{4}$.

\subsubsection{Cut density in regions $R_{D1}$ and $R_{D2}$}

The region $R_{D1}$ is a line segment, with density linearly increasing from the endpoint $(1-b, b)$ (where it is $0$) towards the endpoint $(b,0)$. Therefore, its center of mass is located at $1/3$ of its length, at the point $(c_1,c_2) = \frac13 (1-b,b) + \frac23 (b,0) = (\frac{1+b}{3},\frac{b}{3})$. The cost for this region is $\frac{1-2b}{1-b} \frac{\alpha}{4} (2-c_1-c_2) = \frac{1-2b}{1-b} \frac{\alpha}{4} (2 - \frac13(1+2b))  = \frac{1-2b}{1-b} \frac{\alpha}{12} (5 - 2b)$. We get the same cost for region $R_{D2}$.


\subsubsection{Cut density in regions $R_{E1}$ and $R_{E2}$}
These regions are line segments with a uniform distribution, hence the center of mass is in the middle of the segment which is $(\frac{1+b}{2}, 0)$ in the case of $R_{E1}$ and $(0, \frac{1+b}{2})$ in the case of $R_{E2}$. Therefore, the cost for each region is $\gamma (2 - \frac{1+b}{2}) = \frac{\gamma}{2} (3 - b)$.

\subsubsection{Total cost}

Adding up the costs of all regions, we obtain that the total cost paid by the partition is equal to
\begin{eqnarray*}
\E{cost}  = &  \alpha (2-b) + \frac{\alpha}{6} (4 - b) + \frac{b}{1-b} \frac{\alpha}{2} + \frac{1-2b}{1-b} \frac{\alpha}{6} (5- 2b) \\&~~~+ \gamma (3-b) \\
 = & \frac{\alpha}{6} \left( 6(2-b) + (4-b) + \frac{5-9b+4b^2}{1-b} \right) + \gamma (3-b) \\
 = & \frac{\alpha}{6} \left( 6(2-b) + (4-b) + (5-4b) \right) + \gamma (3-b)  \\
 = & \frac{\alpha}{6} (21 - 11 b) + \gamma (3-b).
\end{eqnarray*} 
Recall that $2 \alpha + 2 \gamma = 1$, hence $\gamma = \frac12 - \alpha$ and the total cost is $\E{cost} = \frac{\alpha}{6} (21 - 11b) + (\frac12 - \alpha) (3-b) = \frac{\alpha}{6} (3 - 5b) + \frac12 (3-b)$. We plug in $b = 2 \sqrt{3} - 3$ which gives $\E{cost} = \frac{10+4\sqrt{3}}{13}$.

\subsection{Performance of threshold cuts}

Here we consider two threshold cuts, say first $\{i: x_{i1} \geq \theta_1\}$ and then $\{i: x_{i2} \geq \theta_2\}$. As we mentioned before, we do not consider the possibility of capturing edges by other terminals here. Let us consider several cases depending on the values of $\theta_1, \theta_2$.

\

{\bf Case 1.} {$\theta_1, \theta_2 \in (0,b]$}

The first threshold cuts the edges in regions $R_A$, $R_{B2}$, $R_{C2}$ and $R_{D2}$. The cost of cutting region $R_A$ is $\frac{\alpha}{b}$. The cost of cutting regions $R_{B2}$, $R_{C2}$ and $R_{D2}$ depends on the value of $\theta_1$ in a linear fashion, and the expected cost across $\theta_1 \in (0,b]$ is $\frac{1}{b} (\frac{\alpha}{4} + \frac{b}{1-b} \frac{\alpha}{4} + \frac{1-2b}{1-b} \frac{\alpha}{4})
 = \frac{1}{b} \frac{\alpha}{2}$. Therefore, the cost varies linearly between $\frac{1}{b} \alpha$ for $\theta_1 = 0$ and cost $0$ for $\theta_1 = b$; more precisely the cost is $\frac{\alpha}{b^2} (b - \theta_1)$.

Since the region $\{\bu: u_1 \geq \theta_1\}$ is allocated to terminal 1, the second threshold $\theta_2$ cuts only edges such that $u_1 < \theta_1$, in particular only in region $R_A$. The cost of this cut is $\frac{\alpha}{b} \cdot \frac{\theta_1}{b} = \frac{\alpha}{b^2}  \theta_1$. Therefore, the combined cost of these two cuts is $\frac{\alpha}{b} + \frac{\alpha}{b^2} (b-\theta_1) + \frac{\alpha}{b^2} \theta_1 = \frac{2\alpha}{b} = \frac{10+4\sqrt{3}}{13}$.

\

{\bf Case 2.} {$\theta_1, \theta_2 \in (b,1]$}

In this case, there is no interaction between the two cuts. The first threshold $\theta_1$ cuts through $R_{B1}, R_{E1}$ and either $R_{B1}$ or $R_{C1}$. Similarly, the second threshold $\theta_2$ cuts through $R_{B2}, R_{E2}$ and either $R_{B2}$ or $R_{C2}$, regardless of the value of $\theta_1$. The cost of both cuts is the same, so we analyze just $\theta_1$.

The cost of cutting region $R_{E1}$ is uniform in the interval $\theta_1 \in [b,1]$, and is equal to $\frac{1}{1-b} \gamma$. We claim that the combined cost of cutting region $R_{B1}$, $R_{C1}$ and $R_{D1}$ is also independent of $\theta_1$: the cost of cutting the triangle $R_{B1}$ increases linearly from $\theta_1 = b$ to $\theta_1 = 1-b$ and then decreases linearly from $\theta_1 = 1-b$ to $\theta_1 = 1$. Conversely, the cost of cutting $R_{D1}$ decreases linearly from $\theta_1 = b$ to $\theta_1 = 1-b$, and then is replaced by the cost of cutting $R_{C1}$ which increases linearly from $\theta_1 = 1-b$ to $\theta_1 = 1$. The probability mass in $R_{D1}$ is $\frac{1-2b}{1-b} \alpha/4$, equal to the portion of the triangle $R_{B1}$ cut by thresholds $\theta_1 \in [b,1-b]$. Similarly, the probability mass in $R_{C1}$ is $\frac{b}{1-b} \alpha/4$, equal to the portion of the triangle $R_{B1}$ cut by thresholds $\theta_1 \in [1-b,1]$. Therefore these contribution balance each other out and add up to a uniform density between $b < \theta_1 \leq 1$, which is the total probability mass in these regions, $\alpha/2$, divided by $1-b$, hence $\frac{1}{2(1-b)} \alpha$. Together with the contribution of $R_{E1}$, this is
$\frac{1}{1-b} \gamma + \frac{1}{2(1-b)} \alpha = \frac{5+2\sqrt{3}}{13}$ after substituting our parameters. Both cuts together have cost $\frac{10+4\sqrt{3}}{13}$.

\

{\bf Case 3.} {$\theta_1 \in (b,1], \theta_2 \in (0,b]$}

In this case, the first cut is analyzed as in Case 2 and has cost $\frac{1}{1-b} \gamma + \frac{1}{2(1-b)} \alpha = \frac{5+2\sqrt{3}}{13}$.
The second threshold cuts at least through the region $R_A$ (and possible some other edges depending on the value of $\theta_1$ but we ignore these). The cost of the second cut is at least $\frac{1}{b} \alpha = \frac{5+2\sqrt{3}}{13}$. Hence the cost of both cuts is at least $\frac{10+4\sqrt{3}}{13}$.

\

{\bf Case 4.} {$\theta_1 \in (0,b], \theta_2 \in (b,1)$}

The cost of the first cut determined by $\theta_1$ is as in Case 1, $\frac{\alpha}{b^2} (b - \theta_1)$. The cost of the second cut is similar to Case 2, but note that now $\theta_2$ cuts only edges with $u_1 \leq \theta_1$, due to the first cut. The reduction in cost compared to Case 2 depends on the value of $\theta_2$: we claim that the cheapest cut is obtained for $\theta_2 = b$. For all other $\theta_2 \in (b,1]$, the cost of the cut is either the same as the one for $\theta_2=b$ (if the point $(\theta_1, \theta_2)$ is inside $R_{B2}$), or the cost is the full cost of Case 2 (if the point $(\theta_1, \theta_2)$ is outside of $R_{B2}$). When $\theta_2 = b$, the cost of the cut is equal to $\frac{1}{1-b} \gamma$ (the cost of cutting $R_{E2}$) plus $\frac{\theta_1}{b} \cdot \frac{1}{2(1-b)} \alpha$, the cost of cutting $R_{B2}$ scaled by the fraction $\frac{u_1}{b}$ of the edges in $R_{B2}$ that are unnasigned to terminal $1$. Thus the total cost of both cuts is
\begin{align*} \frac{\alpha}{b} + \frac{\alpha}{b^2} (b - \theta_1) + \frac{\gamma}{1-b} + \frac{\theta_1}{b} \cdot \frac{1}{2(1-b)} \alpha\\
= \frac{2\alpha}{b} + \frac{\gamma}{1-b} + \left(-\frac{1}{b} + \frac{1}{2(1-b)} \right) \frac{\alpha}{b} \theta_1.  \end{align*}
Since the expression in front of $\frac{\alpha}{b} \theta_1$ is negative, the minimum cost is achieved for $\theta_1 = b$. Then we obtain $\frac{1}{1-b} \gamma + \frac{3}{2(1-b)} \alpha = \frac{10+4\sqrt{3}}{13}$.

\end{document}